\documentclass{llncs}[11]
\usepackage{amsmath}
\usepackage{amsfonts}
\usepackage{amssymb}
\usepackage{ifpdf}
\usepackage{wrapfig}
\usepackage{cite}
\usepackage{clrscode}
\usepackage{comment}
\usepackage{appendix}
\usepackage{import}
\usepackage{subcaption}
\usepackage{sidecap}

\usepackage{caption}
\sidecaptionvpos{figure}{t}

\newcommand{\fractal}[2]{\ensuremath{#1_{#2}}}

\ifpdf

  \usepackage[pdftex]{epsfig}
  \usepackage[pdftex]{hyperref}

\else

    \usepackage[dvips]{epsfig}
    \newcommand{\href}[2]{#2}

\fi

\newif\ifarxiv

\newif\ifabstract
\newif\iffull

\abstracttrue

\ifabstract
	\fullfalse
\else
	\fulltrue
\fi

\newtoks\magicAppendix
\magicAppendix={}
\newtoks\magictoks
\newif\iflater
\laterfalse
\ifabstract
\long\def\later#1{\magictoks={#1}%
  \edef\magictodo{\noexpand\magicAppendix={\the\magicAppendix \par
    \the\magictoks%
  }}
  \magictodo}
\long\def\both#1{\magictoks={#1}%
  \edef\magictodo{\noexpand\magicAppendix={\the\magicAppendix \par
    \noexpand\setcounter{theorem-preserve}{\noexpand\arabic{theorem}}%
    \noexpand\setcounter{theorem}{\arabic{theorem}}%
    \noexpand\setcounter{section-preserve}{\noexpand\arabic{section}}%
    \noexpand\setcounter{section}{\arabic{section}}%
	\noexpand\let\noexpand\oldsection=\noexpand\thesection
	\noexpand\def\noexpand\thesection{\thesection}
	\noexpand\let\noexpand\oldlabel=\noexpand\label
	\noexpand\let\noexpand\label=\noexpand\blank
    \the\magictoks%
    \noexpand\setcounter{theorem}{\noexpand\arabic{theorem-preserve}}%
    \noexpand\setcounter{section}{\noexpand\arabic{section-preserve}}%
	\noexpand\let\noexpand\thesection=\noexpand\oldsection
	\noexpand\let\noexpand\label=\noexpand\oldlabel
  }}
  \magictodo
  \the\magictoks}
\else
\long\def\later#1{#1}
\long\def\both#1{#1}
\fi
\long\def\magicappendix{
	\latertrue%
	\the\magicAppendix%
}

\newtheorem{observation}{Observation}

\newcommand{\Z}{\mathbb{Z}}
\newcommand{\N}{\mathbb{N}}

\newcommand{\dom}{{\rm dom} \;}

\newcommand{\res}[1]{\textrm{res}(#1)}
\newcommand{\termasm}[1]{\mathcal{A}_{\Box}[\mathcal{#1}]}
\newcommand{\prodasm}[1]{\mathcal{A}[\mathcal{#1}]}

\newcommand{\calT}{\mathcal{T}}

\newcommand{\Hfrac}{\bf{H}}
\newcommand{\Ufrac}{\bf{U}}

\begin{document}
\ifabstract
 \addtolength{\belowcaptionskip}{0pt}
 \addtolength{\abovecaptionskip}{-8pt}
\fi

\title{Hierarchical Growth is Necessary and (Sometimes) Sufficient to Self-Assemble Discrete Self-Similar Fractals}

\author{
 Jacob Hendricks%
    \thanks{Department of Computer Science and Information Systems, University of Wisconsin - River Falls, River Falls, WI, USA
    \protect\url{jacob.hendricks@uwrf.edu}}
\and
 Joseph Opseth%
    \thanks{Department of Mathematics, University of Wisconsin - River Falls,
    \protect\url{joseph.opseth@my.uwrf.edu}}
\and
 Matthew J. Patitz
    \thanks{Department of Computer Science and Computer Engineering, University of Arkansas, Fayetteville, AR, USA
    \protect\url{patitz@uark.edu} This author's research was supported in part by National Science Foundation Grants CCF-1422152 and CAREER-1553166.}
\and
    Scott M. Summers\thanks{Computer Science Department, University of Wisconsin--Oshkosh, Oshkosh, WI 54901, USA,\protect\url{summerss@uwosh.edu}.}
}

\institute{}

\date{}

\maketitle

\begin{abstract}
In this paper, we prove that in the abstract Tile Assembly Model (aTAM), an accretion-based model which only allows for a single tile to attach to a growing assembly at each step, there are no tile assembly systems capable of self-assembling the discrete self-similar fractals known as the ``H'' and ``U'' fractals.  We then show that in a related model which allows for hierarchical self-assembly, the 2-Handed Assembly Model (2HAM), there does exist a tile assembly systems which self-assembles the ``U'' fractal and conjecture that the same holds for the ``H'' fractal.  This is the first example of discrete self similar fractals which self-assemble in the 2HAM but not in the aTAM, providing a direct comparison of the models and greater understanding of the power of hierarchical assembly.
\end{abstract}

\vspace{-20pt}
\vspace{-10pt}
\section{Introduction}
\vspace{-10pt}

Systems composed of large, disorganized collections of simple components which autonomously self-assemble into complex structures have been observed in nature, and have also been artificially designed as well as theoretically modeled.  These studies have shown the remarkable power of self-assembling systems to be algorithmically directed across a wide diversity of models with varying dynamics which determine the ways in which the constituent components can combine.  At two ends of an important dimension in this spectrum of dynamics are models in which the atomic components can only combine to growing structures one at a time, e.g. the tile-based abstract Tile Assembly Model (aTAM) \cite{Winf98}, and those in which arbitrarily large assemblies of previously combined components can combine with each other, e.g. the 2-Handed Assembly Model (2HAM) \cite{Versus, CheDot12, Luhrs08, AGKS05g}.  Even though models such as the aTAM which are strictly bound to one-tile-at-a-time growth have been shown to be computationally universal and very powerful in terms of the structures which can self-assemble within them, it has been shown that the hierarchical growth allowed by models such as the 2HAM can afford even greater powers \cite{Versus}.

In pursuit of understanding the boundaries of what is possible in these models, the self-assembly of aperiodic structures has been studied. For example, in~\cite{2HAM-temp1}, a 2HAM system with temperature parameter equal to $1$ is given which self-assembles aperiodic patterns. Aperiodic structures are theoretically fundamental to the concept of Turing universal computation as well as embodied in many mathematical and natural systems as fractals. In fact, the complex aperiodic structure of fractals, as well as their pervasiveness in nature, have inspired much previous work on the self-assembly of fractal structures \cite{FujHarParWinMur07, RoPaWi04}, especially discrete self-similar fractals (DSSF's) \cite{jSADSSF, KL09, TreeFractals, jKautzShutters, RoPaWi04, LutzShutters12, STAM-fractals}. In a tribute to their complex structure, previous work has shown the impossibility of self-assembly of several DSSF's in the aTAM and 2HAM \cite{jSSADST, jSADSSF, KL09, TreeFractals, jKautzShutters, RoPaWi04, LutzShutters12} yet there have also been results showing some models and systems in which their self-assembly is possible \cite{MHAM,STAM-fractals,JonoskaSignals1,JonoskaSignals2}.
Quite notably, a recent result \cite{JacobUCNC2017} is the first to achieve non-scaled self-assembly of a DSSF in the 2HAM.  That work showed that DSSF's with generators (i.e. initial stages which define the shapes of the infinite series of stages) that have square, or 4-sided, boundaries can self-assemble in the 2HAM.  However, they also gave an example of a DSSF with a 3-sided generator that does not.  While previous work has shown sparsely-connected fractals which don't self-assemble in the aTAM or 2HAM  \cite{jSSADST, Versus}, the recent results hinted that perhaps only extremely well-connected fractals, such as those that have 4-sided generators, may be able to self-assemble in the 2HAM, while perhaps none may be able to in the aTAM. In this paper, we continue this line of research into the self-assembly of DSSFs in the aTAM and 2HAM.

In this paper, we specifically consider aTAM and 2HAM systems which finitely self-assemble DSSFs. Finite self-assembly was defined to better understand how 2HAM systems self-assemble infinite shapes (e.g. DSSF's). Intuitively, a shape $S$, finitely self-assembles in a tile assembly system if any finite producible assembly of the system can continue to self-assemble into the shape $S$. Finite self-assembly is a less constrained version of strict self-assembly. Intuitively, a shape $S$ strictly self-assembles in a tile assembly system if it places tiles on -- and only on -- points in $S$. Note that strict self-assembly implies finite self-assembly but the converse is not true in general. For example, a tile system could produce an infinite non-terminal producible assembly that has the property that it cannot self-assemble into the target shape $S$, but any finite producible assembly of the system could self-assemble into $S$.
To further advance the possibility that no DSSF's may self-assemble in the aTAM, we provide impossibility results about fractals with more inter-stage connectivity than any previous fractal whose strict self-assembly in the aTAM was shown to be impossible. In particular, our impossibility results give two fractals which cannot be finitely self-assembled by any aTAM system, which implies that those fractals cannot be strictly self-assembled by any aTAM system either. However, our results also show that the landscape in the 2HAM is more convoluted.  Namely, although \cite{JacobUCNC2017} exhibited a fractal with a 3-sided generator that does not finitely self-assemble in the 2HAM, here we show one which does.  This proves that the boundary between what can and cannot self-assemble in the 2HAM is less understood.  Notably, our impossibility results and constructions are the first to give a head-to-head contrast of the powers of the aTAM and 2HAM to self-assemble DSSF's.  In~\cite{Versus}, shapes are defined which finitely self-assemble in the 2HAM but not in the aTAM, as well as shapes which strictly self-assemble in the aTAM but not in the 2HAM. In this paper, we prove that the hierarchical process of growth attainable in the 2HAM is necessary and sufficient for the self-assembly of certain DSSF's.
Moreover, the construction techniques to build them in the 2HAM do not follow traditional growth patterns of ``stage-by-stage'' growth, but rely fundamentally on combinations of components across a spectrum of hierarchical levels.

\def\latent/{{\texttt{latent}}}
\def\on/{{\texttt{on}}}
\def\off/{{\texttt{off}}}
\def\calF{{\mathcal{F}}}
\vspace{-14pt}
\section{Preliminaries}\label{sec:prelims}
\vspace{-10pt}
Throughout this paper, we use standard definitions of, and terminology related to, the aTAM, 2HAM and discrete self-similar fractals.  For more details of each, please
see Sections~\ref{sec:atam-informal} and \ref{sec:2ham-informal}.
In this section, we include only the few definitions unique to this paper.
\vspace{-13pt}
\subsection{Definitions for the aTAM and 2HAM}
\vspace{-8pt}
Let $\vec{\alpha}$ be an assembly sequence of an aTAM system. In the following, $\vec{\alpha}[i]$ denotes the tile that $\vec{\alpha}$ places at assembly step $i$. We say that $\vec{\alpha}[i]$ is the \emph{parent} of $\vec{\alpha}[j]$ if $i < j$ and $\vec{\alpha}[j]$ binds to $\vec{\alpha}[i]$. Furthermore, we say that tile $\vec{\alpha}[i]$ is the \emph{ancestor} of a tile $\vec{\alpha}[k]$ if either $\vec{\alpha}[i]$ is the parent of $\vec{\alpha}[k]$, or there exists an index $j$, such that, $i < j < k$, $\vec{\alpha}[j]$ is the parent of $\vec{\alpha}[k]$ and $\vec{\alpha}[i]$ is the ancestor of $\vec{\alpha}[j]$. Note that $\vec{\alpha}[j]$ implicitly refers to both the tile type and location, and the parent and ancestor relationships, in general, depend on the given assembly sequence $\vec{\alpha}$.

For an infinite shape $X\subseteq \Z^2$ and an aTAM or 2HAM system $\calT$, we say that $\calT$ \emph{finitely self-assembles} $X$ if every finite producible assembly of $\calT$ has a possible way of growing into an assembly that places tiles exactly on those points in $X$. In this paper we consider finite self-assembly of DSSF's (in the strict sense).

\vspace{-12pt}
\subsection{The U-fractal and H-fractal}
\vspace{-2pt}
For the definition of discrete self-similar fractal (DSSF)\footnote{Note that we use the standard DSSF definition in which DSSF's are contained within quadrant $I$ of $\mathbb{N}^2$.  However, our impossibility result proofs could be trivially modified to hold for alternate definitions which allow for DSSFs to occupy any set of quadrants.}
see Section~\ref{sec:dssf-def}.
\vspace{-8pt}

\begin{definition}
The \emph{$\Ufrac$ fractal} is the DSSF whose generator consists of exactly the points $\{(0,0),(0,1),(0,2),(1,0),(2,0),(2,1),(2,2)\}$.
\end{definition}

\begin{definition}
The \emph{$\Hfrac$ fractal} is the DSSF whose generator consists of exactly the points $\{(0,0),(0,1),(0,2),(1,1),(2,0),(2,1),(2,2)\}$.
\end{definition}

\vspace{-16pt}
\section{Brief proof of the impossibility of finite self-assembly of the $\Hfrac$ fractal in the aTAM}\label{sec:aTAM-brief-impossible}
\vspace{-4pt}

The $\Hfrac$ fractal is defined as shown in Figure~\ref{fig:H}.  Let $h_i$ be the $i$-th stage of $\Hfrac$. We call the \emph{center} tile of $h_i$, denoted as $center(h_i)$, the tile in the center of the stage that connects the \emph{left} and \emph{right} halves of $h_i$.

Let $B^{\Hfrac}_0 = \{(0,0),(0,1),(0,2),(2,0),(2,1),(2,2)\}$.  For stages $i > 1$, we call the following set of 6 points the \emph{bottleneck} points of $h_i$, or $B^{\Hfrac}_i$: \\ $B^{\Hfrac}_i = \left\{\left.\left(3^{i-1}+\frac{3^{i-2}-1}{2},3^{i-1}+\frac{3^{i-2}-1}{2}\right) + 3^{i-2}b \; \right| \; b \in B^{\Hfrac}_0 \right\}$. An example of the bottleneck points for a few stages of $\Hfrac$ can be seen in Figure~\ref{fig:H}. In what follows, we will use the term ``bottleneck tile'' to refer to the tile placed (by some assembly sequence) at that bottleneck point.

\begin{wrapfigure}{r}{0.5\textwidth}
\begin{center}
\vspace{-35pt}
\includegraphics[width=2.0in]{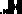}
\caption{First three stages of the $\Hfrac$ fractal, with the left-most being the generator.  The bottleneck points of stages 2 and 3 (blue).\vspace{-35pt}}
\label{fig:H}

\end{center}
\end{wrapfigure}

The top, middle and bottom bottleneck points of $h_i$ are denoted as $top(i)$, $middle(i)$ and $bottom(i)$. We will refer to the points in $h_i$ in between its center tile and left bottleneck points as its \emph{left-center}. Assuming $\Hfrac$ finitely self-assembles in some TAS $\mathcal{T}$, then every tile placed in the left-center of $h_i$, for all $i \ge j$ for some $j \in \mathbb{N}$, has as an ancestor, relative to some $\mathcal{T}$ assembly sequence $\vec{\alpha}$, at least one bottleneck point. We call a tile in the left-center of $h_i$ \emph{top-left-placed} if $top(i)$ is its ancestor and $middle(i)$ and $bottom(i)$ are not its ancestors. We define \emph{middle-left-placed} and \emph{bottom-left-placed} tiles (in the left-center of $h_i$) similarly. Note that, if the parent of the center tile of $h_i$ is adjacent to the left, then every tile in the left-center of $h_i$ must have some bottleneck point (either top, middle or bottom) in the left half of $h_i$ as an ancestor.

\begin{theorem}\label{thm:H-impossible}
$\Hfrac$ does not finitely self-assemble in the aTAM.
\end{theorem}

\begin{proof}
For the sake of obtaining a contradiction, assume there exists an aTAM TAS $\calT = (T,\sigma,\tau)$ in which $\Hfrac$ finitely self-assembles.  We will show that $\Hfrac$ does not finitely self-assemble in $\mathcal{T}$.  Without loss of generality, we will assume that $|\sigma| = 1$, i.e. that $\calT$ is singly-seeded but our proof technique will hold for any TAS $\mathcal{T}$ with finite seed assembly. Since the location of $\sigma$ must be within $\Hfrac$, let $s$ be the stage number of the smallest stage of $\Hfrac$ which contains $\sigma$.

Let $c = 6|T|^6$.  If $\Hfrac$ finitely self-assembles in $\calT$, then every producible assembly in $\calT$ has domain contained in $\Hfrac$. Let $\vec{\alpha}$ be the shortest assembly sequence in $\calT$ whose result has domain $h_{c+s+2}$, subject to the additional constraint that, when multiple locations could receive a tile in a given step, $\vec{\alpha}$ always places a tile in a location of the smallest possible stage.

By our choice of $c$, we know that there are at least 6 stages of $\Hfrac$ whose respective bottleneck points are identically tiled by $\vec{\alpha}$. Since, in any assembly sequence, the center tile of each stage of $\Hfrac$ either has a parent adjacent to the left or right, it follows, without loss of generality, that there are at least 3 stages, namely $h_i$, $h_j$ and $h_k$, for $i < j < k$, whose respective bottleneck points are identically tiled by $\vec{\alpha}$ and whose respective center tiles have parents adjacent to the left.

Relative to $\vec{\alpha}$, there are three cases to consider:
(1) and (2) some top-left-placed (bottom-left-placed) tile of the left-center of $h_j$ is placed at a point that is not contained in an $h_{j-3}$, appropriately-translated, so that $center(j-3)$, appropriately-translated, is $top(j)$ ($bottom(j)$), or
(3) some middle-left-placed tile of the left-center of $h_j$ is placed at a point that is not contained in an $h_{j-2}$, appropriately-translated, so that $center(j-2)$, appropriately-translated, is $top(j)$. (Intuitively, these are conditions specifying how far growth from each bottleneck tile extends toward its neighbors before utilizing cooperation with growth from them.)
Note that, if none of these cases apply, then the left-center of $h_j$ wouldn't assemble completely and $\Hfrac$ wouldn't finitely self-assemble in $\mathcal{T}$.

{\bf Case 1:} Use $\vec{\alpha}$ to create a new valid assembly sequence in $\mathcal{T}$ as follows.  Starting from the seed, run $\vec{\alpha}$ until the step at which it places the first bottleneck tile on the left side of $h_j$. Then, begin recording a sub-sequence of $\vec{\alpha}$ and denote this sub-sequence as $\vec{\alpha}'$.  As we run $\vec{\alpha}$ forward from this point, until it places the last tile of $h_j$, whenever a top-left-placed tile in $h_j$ is placed by $\vec{\alpha}$, we add that tile placement (type and location) to $\vec{\alpha}'$.  In this way, $\vec{\alpha}'$ becomes a sub-sequence of $\vec{\alpha}$ that records the growth of the top-placed sub-assembly -- and only the top-placed sub-assembly -- of the left-center of $h_j$.

Now, reset $\vec{\alpha}$ to the seed and begin its forward growth until the placement of the first bottleneck tile on the left side of $h_i$ (recall $i < j$).  At this point, merge $\vec{\alpha}$ and $\vec{\alpha}'$ as follows.  For each tile position $\vec{p}$ in $\vec{\alpha}'$, we translate it so that the new position, $\vec{p}'$, is the point with the same relative offset from the top-left bottleneck position of $h_i$ as $\vec{p}$ was from the top-left bottleneck position of $h_j$. Continue to run $\vec{\alpha}$ forward by performing all tile placements up to, and including, the placement of $top(i)$, with the exception of the $middle(i)$, $bottom(i)$, or any descendants thereof. As soon as $\vec{\alpha}$ places $top(i)$, we follow the tile placements of the modified $\vec{\alpha}'$. The result is a valid assembly sequence up to the point of the placement of at least one tile outside of $\Hfrac$ (since the portion of the left-center of $h_j$ grown by $\vec{\alpha}'$ doesn't fit within the locations of $\Hfrac$ available in $h_i$). Thus, $\Hfrac$ does not finitely self-assemble in $\calT$. A similar scenario, but for a different fractal, in which such out-of-bounds growth may occur, is depicted in Figure~\ref{fig:main-U5-bad-growth-in-3-and-4}. \\
{\bf Case 2:}  This case is symmetric to the previous case. \\
{\bf Case 3:}  First, create an assembly sub-sequence $\vec{\alpha}''$ that records the tile placements of only the middle-placed tiles of $h_j$, similar to the construction of $\vec{\alpha}'$ in Case 1. Then, run $\vec{\alpha}$ forward, starting from the seed, performing all tile placements up to, and including, the placement of the $middle(k)$, with the exception of $top(k)$ or $bottom(k)$, or descendants thereof. As soon as $\vec{\alpha}$ places $middle(k)$, we follow the tile placements of the modified $\vec{\alpha}''$, appropriately-translated, from $h_j$ to $h_k$. Here, we are essentially replaying the assembly of a smaller stage within a larger stage. The result is a valid assembly sequence up to the point of the placement of at least one tile outside of $\Hfrac$ (due to the specifically different scales of portions of $\Hfrac$ in $h_j$ and $h_k$). Thus, $\Hfrac$ does not finitely self-assemble in $\calT$.
\qed

\end{proof}

\begin{corollary}\label{cor:H-strict-impossible}
$\Hfrac$ does not strictly self-assemble in the aTAM.
\end{corollary}

Since strict self-assembly of a shape $S$ by a system $\mathcal{T}$ implies finite self-assembly of $S$ by $\mathcal{T}$, Corollary~\ref{cor:H-strict-impossible} follows from Theorem~\ref{thm:H-impossible}. 

\vspace{-16pt}
\section{Impossibility of Finite Self-Assembly of the $\Ufrac$ fractal in the aTAM}\label{sec:u-impossible}
\vspace{-6pt}
The $\Ufrac$ fractal is defined as shown in Figure~\ref{fig:U}.

\begin{theorem}\label{thm:U-impossible}
$\Ufrac$ does not finitely self-assemble in the aTAM.
\end{theorem}

Due to space constraints, we only give brief description of the proof of Theorem~\ref{thm:U-impossible}.  Essentially, the proof is very similar to that of Theorem~\ref{thm:H-impossible}.  $\Ufrac$ has bottlenecks (which can be seen in Figure~\ref{fig:U}) similar to $\Hfrac$, and in a similar way, it is impossible for the portion of stages inside of the bottlenecks to self-assemble since the tiles at bottleneck locations of multiple stages must be identical, and growth which would have to be possible within one stage would be able to grow out of bounds of $\Ufrac$ in a different stage.  An example can be seen in Figure~\ref{fig:main-U5-bad-growth-in-3-and-4}, and more details of the proof can be found in
Section~\ref{sec:u-appendix}.

\begin{figure}[htp!]
\vspace{-16pt}
\centering
   \begin{subfigure}[t]{1.5in}
\centering
   \includegraphics[width=1.4in]{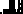}
   \caption{\label{fig:U}}
    \end{subfigure}
    \quad\quad
      \begin{subfigure}[t]{3.0in}
		\centering
   \includegraphics[width=2.5in]{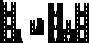}
   \caption{\label{fig:main-U5-bad-growth-in-3-and-4}}
    \end{subfigure}
   \caption{(a) First three stages of the $\Ufrac$ fractal, with the leftmost being the generator.  The bottleneck points of stages 2 and 3 are colored blue. (b) Depiction of how top-placed growth from stage 5 would go out of bounds of $\Ufrac$ in stage 3 and stage 4.  (left) A portion of stage 5 showing the 3 bottleneck tiles in black, and possible horizontal and vertical growth from the top bottleneck tile.  (middle and right)  Stages 3 and 4.  The black tile is the top left bottleneck tile, the green locations are those which correctly match the smaller stage, and the red are those which go out of bounds of $\Ufrac$. Clearly, all tiles in green positions will be able to grow, and then erroneous growth is forced to occur immediately east of the green tiles, where no other tiles could prevent this growth. (Note that only a single tile needs to be placed in a red location to break the shape of $\Ufrac$.)\vspace{-24pt}}
   \label{fig:U-stuff}
\end{figure}

\begin{corollary}\label{cor:U-strict-impossible}
$\Ufrac$ does not strictly self-assemble in the aTAM.
\end{corollary}

\vspace{-16pt}
\section{U-fractal Finitely Self-assembles in the 2HAM}\label{sec:2HAM-construction}
\vspace{-6pt}

In this section we show how to finitely self-assemble the $U$-fractal, $\bf{U}$, DSSF in the 2HAM (with scale factor of $1$) at temperature $2$. We will present our construction under the assumption that a particular assembly sequence is followed. We then show that the construction also holds for an arbitrary choice of assembly sequence. Here, we present the main idea of the construction and give more detail in 
Section~\ref{sec:construction-appendix}. 
First, we state our main positive result.

\begin{theorem}\label{thm:U-fractal}
Let $\bf{U}$ be the U-fractal DSSF. There exists a 2HAM TAS $\calT_{\bf{U}} = (T_{\bf{U}}, 2)$ that finitely self-assembles $\bf{U}$.
Moreover,  $\calT_{\bf{U}}$ has the property that for every stage $s\geq1$ and every terminal assembly $\alpha \in \termasm{T_{\bf{U}}}$, $\fractal{U}{s} \subset \dom(\alpha)$ (modulo translation).
\end{theorem}
\vspace{-2pt}

We now introduce notation useful for describing the sets of points (including singleton sets) in a fractal. We start with a notation for the \emph{address} of a point in a stage $U_n$ of $\mathbf{U}$. Figure~\ref{fig:addresses} describes this notation for $U_3$. Similar notation for $U_n$ is defined recursively.

\begin{figure}
\vspace{-10pt}
    \centering
    \begin{minipage}[t]{0.70\textwidth}
        \centering
\includegraphics[width=1.9in]{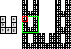}
\caption{(left) Address labels of each point in the generator of $\Ufrac$, (right) The black location is contained within stage three, and its address is $dab$ (i.e. it is location $d$ in a stage one copy (outlined in red), within location $a$ of a stage two copy (outlined in green), within location $b$ of stage three.) }
\label{fig:addresses}
    \end{minipage}
    \hfill
    ~ %
    \begin{minipage}[t]{0.25\textwidth}
        \centering
\includegraphics[width=1.2in]{images/U3}
\caption{The set of dark gray points of $U_3$ are referred to as a stage-$2$ ladder.}
\label{fig:stage-2-ladder}
    \end{minipage}
    \vspace{-16pt}
\end{figure}
The address of a point in $U_n$ is a string of $n$ symbols of $\{a, b, c, d, e, f, g\}$. Therefore, to define a subset, $S$ say, of points in $U_n$, it is convenient to use regular expressions to describe the strings corresponding to addresses of points in $S$.
Figure~\ref{fig:stage-2-ladder} depicts a set of points in $U_3$ which we refer to as a \emph{stage-$2$ ladder}. This set is defined by the regular expression $[defg][abc][ab]|$ $[abcdefg]d[ab]|$ $[abcd][efg][ab] |$ $ [defg][ab]c |$ $ [ef]cc | [abef]dc |$ $ [abcd][ef]c | [ab]gc$.

We also introduce terminology for some of the more important shapes that the 2HAM system which self-assembles $\bf{U}$ self-assembles. These shapes are \emph{stage-$n$ ladders}, \emph{left rungs}, and \emph{right rungs}. Figure~\ref{fig:2ladder} depicts a stage-$2$ ladder. The two rightmost supertiles in Figure~\ref{fig:stage2rungs} depict left and right rungs where the rightmost supertile is a right rung. Let $S_n$ by the set of points in $U_{n+1}$ with addresses given by the expression $.\{n\}[abc]$ (i.e. strings of length $n+1$ ending in $a$, $b$, or $c$. In other words, $S_n$ is $\{(x, y + m*3^n) | (x,y)\in U_n, m\in \{0, 1, 2\} \}$. Also let $B$ be the set of westernmost, easternmost, and sothernmost points of $S_n$. Then, a stage-$n$ ladder is the shape defined to be the points in $S_n \setminus B$. Figure~\ref{fig:2ladder} (right) depicts a supertile with the shape of a stage-$3$ ladder. We are now ready to present the construction which shows Theorem~\ref{thm:U-fractal}.

\subsection{$U$-fractal construction overview}\label{sec:ufractal}

\begin{figure}[htp]
\begin{center}
\includegraphics[width=1.7in]{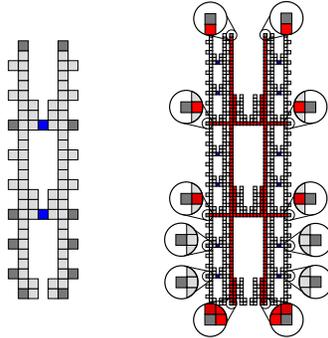}
\caption{A depiction of a stage-2 ladder (left) and a stage-3 ladder (right). Dark gray squares denote tile locations where tiles may contain an edge that has a special glue called an ``indicating glue''. The goal of the construction is to define a 2HAM system that 1) self-assembles $10$ types of stage-2 ladder supertiles (the type of a stage-$2$ ladder supertile depends on whether or not tiles at dark gray locations contain indicating glues), and 2) for $n\geq 3$ self-assembles $10$ types of stage-$n$ ladder supertiles from stage-$(n-1)$ ladder supertiles such that the stage-$n$ ladder supertile contains tiles that have indicating glues (at locations shown in dark gray locations in the figure on the right for stage-$3$ ladder supertiles).}
\label{fig:2ladder}
\end{center}
\end{figure}

In this section, we describe a 2HAM system that finitely self-assembles U. We do this by describing the supertiles producible in the 2HAM system and note that tiles can be defined so that these supertiles self-assemble.
We first describe \emph{base supertiles} that initially self-assemble and then describe how these base supertiles can bind to self-assemble supertiles that contain larger and larger stages of $\bf{U}$. In all, the supertiles which self-assemble in $\calT_{\bf{U}}$ are as follows.
\begin{enumerate}
\item $12$ different types of base supertiles that are hard-coded to self-assemble, $10$ of which have the shape of a stage-$2$ ladder, and $2$ of which have the shape of either a left or right rung. We call these supertiles stage-$2$ ladder supertiles and left or right rung supertiles respectively. Figures~\ref{fig:Ubase} and~\ref{fig:stage2rungs} (left two supertiles) depict the $10$ different stage-$2$ ladder supertiles. The two righmost supertiles shown in Figure~\ref{fig:stage2rungs} are left and right rung supertiles.
\item For each $n$, $12$ different types of supertiles self-assemble which have the shape of a stage-$n$ ladder. We call these supertiles stage-$n$ ladder supertiles. Figure~\ref{fig:2ladder} (right) shows a stage-$3$ ladder supertile.
\item Supertiles which we refer to as \emph{grout} supertiles are hard-coded to bind to stage-$n$ ladders for any $n\in \N$. For all $n \geq 2$, grout supertiles bind to stage-$n$ ladders (and also bind to left and right rungs as a special case) to yield supertiles that expose glues which bind in some assembly sequence to yield a stage-$(n+1)$ ladder. Figure~\ref{fig:grouted-stage2} depicts $6$ stage-$2$ ladders and $6$ stage-$2$ rungs with grout supertiles attached. We refer to a stage-$n$ ladder supertile (resp. rung supertile) with grout supertiles attached such that no more grout supertiles can attach as a \emph{grouted} stage-$n$ ladder supertiles (resp. rung supertile). Finally, grout supertiles that bind to stage-$n$ ladders are referred to as ``grout for stage-$(n+1)$''. As we will see there are $10$ different types of grout corresponding to the $10$ different types of stage-$2$ ladder supertiles.
\end{enumerate}

\begin{wrapfigure}{r}{0.4\textwidth}
\begin{center}
\vspace{-22pt}
\includegraphics[width=1in]{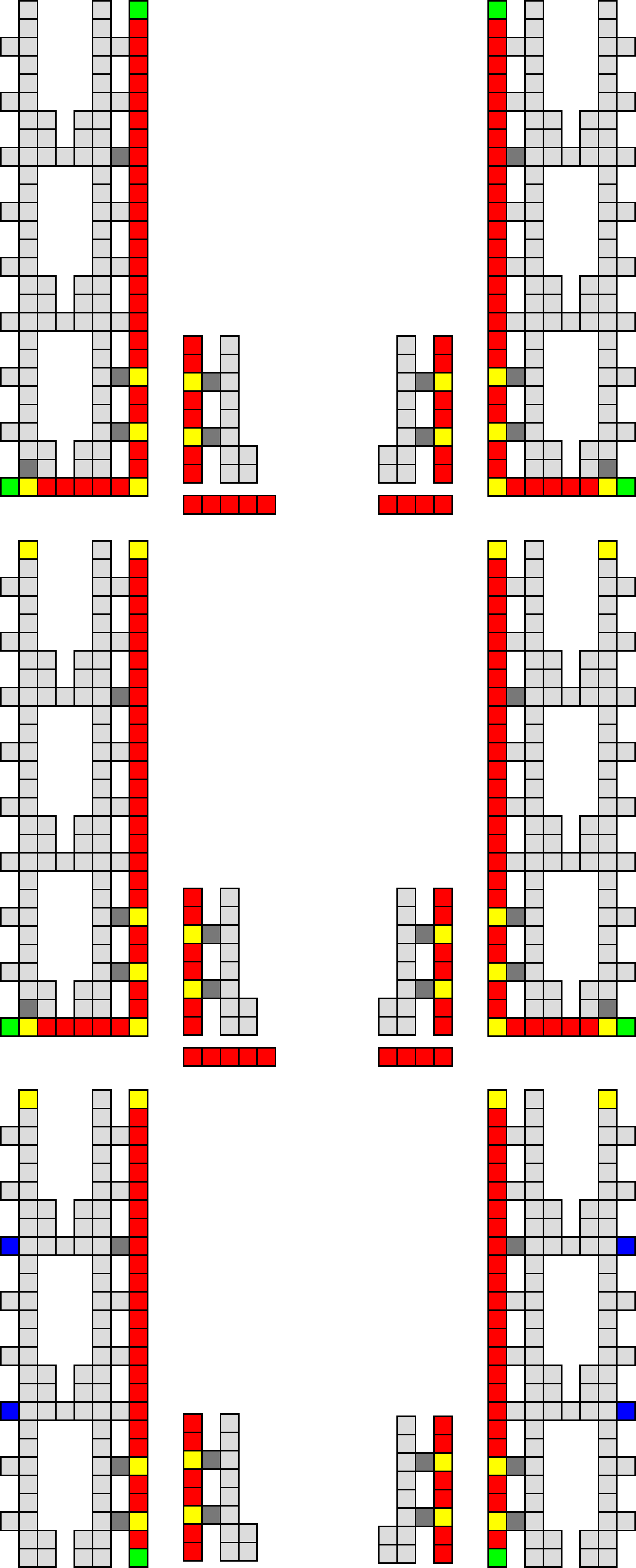}
\caption{A schematic depiction of grouted stage-$2$ ladder supertiles and grouted rung supertiles. There are $6$ types of ladder supertiles shown here. Tiles shown as yellow squares contain strength-$1$ glues which we call ``binding glues'' that allow the depicted grouted ladder supertiles to bind. Tiles shown as green or blue squares may contain edges with indicating glues and whether or not an indicating glues is on an edge of a tile at a green or blue location depends on which of the $10$ typegs of grout that binds (i.e. which type of stage-$3$ verson of a stage-$2$ ladder supertile is self-assembling.) Note that tiles in locations shown as blue squares are contained in a stage-$2$ ladder supertile.}
\label{fig:grouted-stage2}
\vspace{-40pt}
\end{center}
\end{wrapfigure}

Throughout this section we describe the self-assembly of the above supertiles by describing a particular assembly sequence. We note that there are many other assembly sequences for $\calT_{\bf{U}}$ and many possible producible supertiles. This is due to the fact that proper subassemblies of the supertiles described above are themselves producible. Nevertheless, we show that this nondeterminism does not prevent $\bf{U}$ from being finitely self-assembled. For now, we consider assembly sequences such that for $n\geq 3$, 1) stage-$(n-1)$ ladder supertiles completely self-assemble before grout supertiles for stage-$n$ bind, 2) grout for stage-$n$ binds to stage-$(n-1)$ ladder supertiles until a grouted stage-$n$ ladder supertile self-assembles (i.e. grout supertiles bind to stage-$(n-1)$ ladder supertiles until no other grout supertiles can bind), and 3) stage-$n$ ladder supertiles self-assemble from grouted ladder supertiles of previous stages. Figure~\ref{fig:sequence} depicts such an assembly sequence for $n=3$. Note that grout supertiles bind to completed stage-$2$ ladder and rung supertiles before the stage-$3$ ladder self-assembles.

Referring to Figure~\ref{fig:2ladder}, the main idea behind the construction is to defined a tile set which self-assembles base supertiles and grout supertiles. Grout supertiles bind to base supertiles to yield supertiles which in turn bind to yield stage-$3$ ladder and rung supertiles. In particular, the stage-$3$ ladder and rung supertiles which self-assemble are analogous to (i.e. are higher stage versions of) stage-$2$ base and rung supertiles. See Figure~\ref{fig:2ladder} (left and right) for more detail.
We now describe base and grout supertiles, the tiles that self-assemble them, as well as the assembly sequences for these supertiles and higher stages of $\bf{U}$ in more detail.

\begin{figure}[htp!]
\vspace{-15pt}
\centering
	   \includegraphics[width=4.5in]{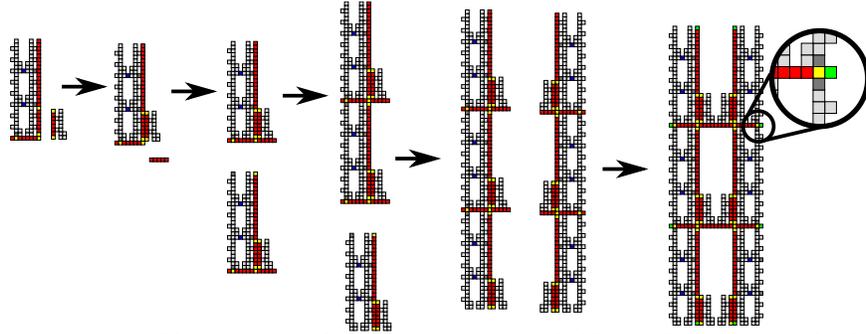}
   \caption{An assembly sequence where grouted stage-$2$ ladder and rung supertiles bind to yield a stage-$3$ ladder supertile. Note that the result of this assembly sequence is a stage-$3$ ladder supertile.\vspace{-0pt}} \label{fig:sequence}
\end{figure}

\vspace{-16pt}
\subsubsection{The $12$ base-supertiles}\label{sec:base-supertiles}

The tile set which initially self-assembles stage-$2$ ladder supertiles and rung supertiles are defined so that these supertiles contain tiles that expose special glues in specific locations; possible locations for special glues are shown in dark gray in the Figures~\ref{fig:Ubase} and~\ref{fig:stage2rungs}. We call these special glues \emph{indicating glues}. The purpose of indicating glues will be described in Section~\ref{sec:grout-supertiles}. In this section we describe the $12$ different types of base supertiles, starting with the $10$ stage-$2$ ladder supertiles.

\vspace{-10pt}
\begin{SCfigure}

\includegraphics[scale=.028]{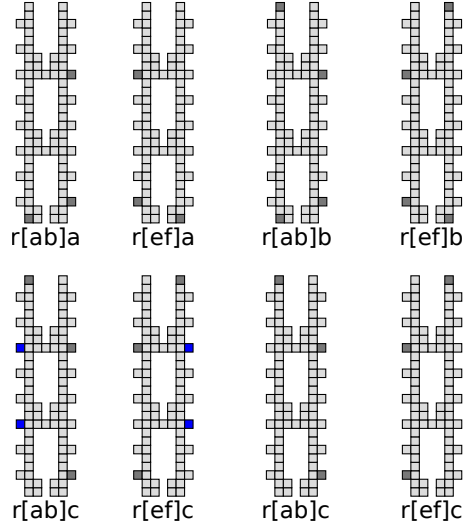}
\caption{(Right) A depiction of $8$ types of the stage-$2$ ladder supertiles. Each of the $8$ figures is labeled with a regular expression defining the set of points in $U_4$ where $r = (r1|r2)$ such that $r1 = [defg][abc]|[abcdefg]d|[abcd][efg]$ and $r2 = $\\ $[defg][ab]|[ef][cd]|[ab][dg]|[abcd][ab]$. The label also describes  where these stage-$2$ supertiles will be located within a stage-$3$ ladder supertile (the tile locations of which are a subset of $U_4$). We will use these labels to refer to a stage-$2$ ladder supertile type. We also note that there are two versions of stage-$2$ ladder supertiles with type $r[ab]c$ and two versions with type with type $r[ef]c$.}
\label{fig:Ubase}
\end{SCfigure}
\vspace{-20pt}

\begin{wrapfigure}{r}{0.5\textwidth}
\vspace{-12pt}
\centering
	   \includegraphics[scale=.028]{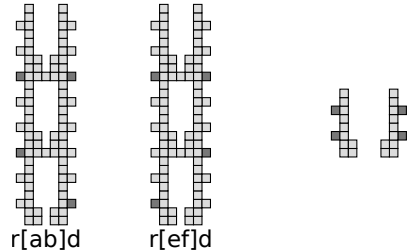}
   \caption{A depiction of $2$ stage-$2$ ladder supertiles labeled using the same scheme as described in Figure~\ref{fig:Ubase} (right) and  a depiction of stage-$2$ left and right rungs (right). The rightmost supertile is the right rung.} \label{fig:stage2rungs}
   \vspace{-24pt}
\end{wrapfigure}

Stage-$2$ ladder supertiles are hard-coded to self-assemble via particular assembly sequences described in Figure~\ref{fig:sequence-stage2ladder}. As we will see, enforcing such assembly sequences will help ensure proper self-assembly of consecutive ladder stages. For now, we assume that the stage-$2$ ladder supertiles completely self-assemble prior to binding to supertiles to yield larger assemblies. Tile types are defined so that $10$ different types of stage-$2$ ladder supertiles that self-assemble.
Referring to the stage-$2$ ladder supertiles in Figures~\ref{fig:Ubase} and ~\ref{fig:stage2rungs}, tiles can be hard-coded so that edges of tiles shown as dark gray squares expose indicating glues. The type of a stage-$2$ ladder supertile is uniquely determined by the locations and types of indicating glues on edges of the tiles that it contains. Moreover, for each base supertile, all of the indicating glues are distinct.
We note that a stage-$2$ ladder supertile's type also determines its location as a subassembly of a stage-$3$ ladder supertile.

\begin{wrapfigure}{r}{0.44\textwidth}
\vspace{-20pt}
\centering
	   \includegraphics[width=1.5in]{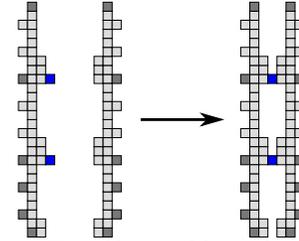}
   \caption{To self-assemble each stage-$2$ ladder supertile, glues for each of the tiles in the supertile are hard-coded. In particular, the abutting edges of tiles at locations corresponding to each square of the left and middle supertiles shown here contain matching strength-$2$ glues and each such glue is unique for each base supertile. Tiles shown as blue squares of a stage-$2$ ladder supertile have strength-$2$ glues on their west edges and strength-$1$ glues on their east edge. This ensures that the ``left half'' (left) and ``right half'' (middle) (or portions of each) sufficiently self-assemble before each half binds. \vspace{-28pt} } \label{fig:sequence-stage2ladder}
\end{wrapfigure}

Except for tiles containing indicating glues, the non-abutting north (respectively south, east, and west) edges of
northernmost (respectively southernmost, easternmost and westernmost) tiles of complete stage-$2$ ladders contain strength-$1$ glues, all with the same glue type which we label $n$ ($s$, $e$, and $w$ respectively). We call such glues \emph{generic glues}. Generic glues are not shown in figures. The purpose of these glues is to facilitate the binding of grout supertiles as such supertiles bind to yield grouted stage-$2$ ladder supertiles. For each of the $10$ types of stage-$2$ ladder supertiles, tiles at locations depicted by gray squares in Figure~\ref{fig:Ubase} contain indicating glues (the purpose of which we describe in more detail next). Finally, in addition to stage-$2$ ladder supertiles, tiles are hard-coded so that left and right rungs self-assemble. These supertiles also contain indicating glues at tiles with locations shown as gray squares in Figure~\ref{fig:stage2rungs} (two leftmost figures). We next describe grout supertiles.

\vspace{-12pt}
\subsubsection{Grout supertiles}\label{sec:grout-supertiles}

There are $10$ different types of grout supertiles corresponding to the $10$ different types of stage-$2$ ladder supertiles. Intuitively, grout binds to ladder supertiles to yield grouted ladder supertiles. For $n\geq 3$, appropriate grouted ladder supertiles with stage less than $n$ bind to yield a stage-$n$ ladder supertile. The resulting stage-$n$ ladder supertile will contain tiles with edges that contain glues identical to the indicating glues of one of the $10$ types of stage-$2$ ladder supertiles. Therefore, the indicating glues of edges of tiles of a stage-$n$ ladder supertile determine the \emph{type} for the stage-$n$ ladder supertile. The type of stage-$n$ ladder supertile that results is determined by the type of grout that binds to the ladder supertiles with stage less than $n$ that bind to yield the stage-$n$ ladder supertile. Figure~\ref{fig:grouted-stage2} shows $6$ different types of stage-$2$ ladder supertiles bound to grout supertiles (shown in red, green, and yellow). The $4$ types of stage-$2$ ladder supertiles not shown in Figure~\ref{fig:grouted-stage2} only bind during the self-assembly of a stage-$n$ ladder supertile for $n \geq 4$. Figure~\ref{fig:grouted-stage2} also shows stage-$2$ left and right rungs that are bound to grout as well as grout supertiles which consist only of red tiles. Tiles belonging to supertiles depicted in Figure~\ref{fig:grouted-stage2} as yellow tiles expose binding glues which allow for the binding of these supertiles. The locations of these yellow tiles are determined by the indicating glues of the stage-$2$ ladder supertiles. We next describe the grout supertiles that bind and how they bind to $3$ types of stage-$2$ ladder supertiles. The grout supertiles that bind and how they bind to the other types of stage-$2$ ladder supertiles is similar.

\begin{figure}[htp!]
\vspace{-24pt}
\centering
	   \includegraphics[width=3in]{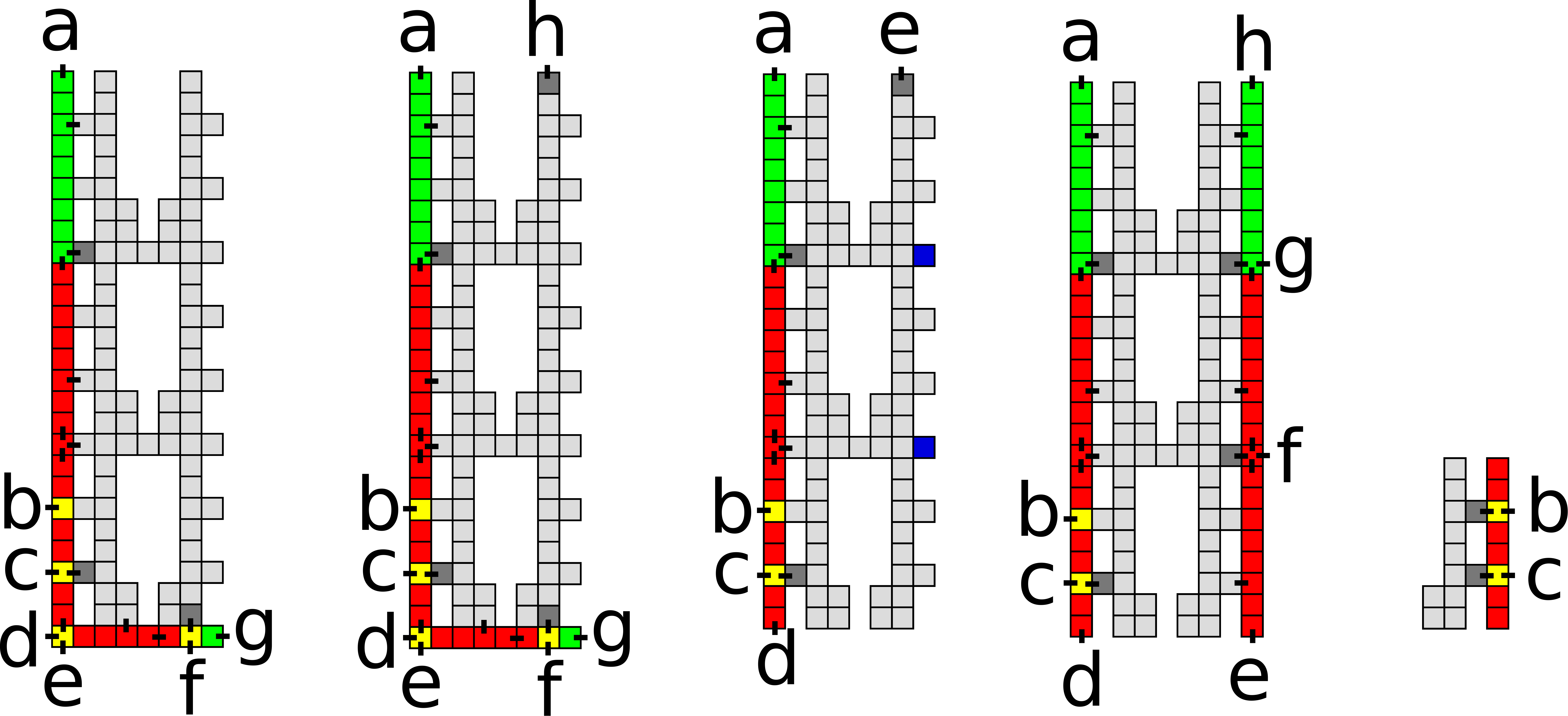}
   \caption{A schematic depiction of $5$ supertiles. From left to right, the first supertile is a grouted stage-$2$ ladder supertile with type $r[ef]a$, the next supertile is a grouted stage-$2$ ladder supertile with type $r[ef]b$, the next supertile is a grouted stage-$2$ ladder supertile with type $r[ef]c$, the next supertile is a grouted stage-$2$ with type $r[ef]d$, and the last supertile is a a grouted right rung supertile. Glue labels shown here are for reference purposes only and do not correspond to the label in the definition of the tile set for $\calT_{\bf{U}}$. Note that many of the glues of these supertiles are not depicted and the bound strength-$1$ glues shown here are intended to indicate how the grout supertiles cooperatively bind.\vspace{-20pt}} \label{fig:grout-sequence1}
   \vspace{-10pt}
\end{figure}

Like stage-$2$ ladder and rung supertiles, grout supertiles are hard-coded to self-assemble and there are $10$ different types of grout supertiles which self-assemble. We describe the grout supertiles which bind to the stage-$2$ ladder supertiles with types $r[ef]a$, $r[ef]b$, $r[ef]c$, and $r[ef]d$.
Let $L$ be a stage-$2$ ladder supertile with type $r[ef]b$. We denote as $L'$ the supertile that is the result of grout binding to $L$ until no more grout supertiles can bind.  We refer to the labels for the glues shown in the second figure from the left in Figure~\ref{fig:grout-sequence1}. First, we note that the grout supertiles shown with green tiles initially binds. The abutting edges of this supertile with no glues shown in the figure have strength-$2$ glues that hard-code the self-assembly of this supertile. This is also the case with the other grout supertiles shown in Figure~\ref{fig:grout-sequence1}. Note that grout supertiles are defined to cooperatively bind to $L$ to partially surround this supertile. We now describe the glues labeled $a$ through $h$. The glue labeled $a$ is a strength-$1$ glue that encodes the type of grout that binds to $L$. The glue labeled $h$ is a non-generic ``helper'' glue. Together $a$ and $h$ cooperate to permit the binding of $L'$ to a grouted stage-$2$ supertile with type $r[ef]a$, $B$ say, iff the grout types of $L'$ and $B$ are the same. The glues $b$ and $c$ belong to a grout supertile that only ever binds stage-$2$ ladder supertiles; this can be enforced by the definition of the tile types which self-assemble grout supertiles. $b$ and $c$ do not encode the grout type as this is not necessary for the construction, but they do allow for a grouted right rung supertile (such as the one depicted in the rightmost figure of Figure~\ref{fig:grout-sequence1}) to bind. As shown in Figure~\ref{fig:sequence}, this is important for allowing stage-$3$ ladder supertile to self-assemble from $L'$.
\begin{wrapfigure}{r}{0.3\textwidth}
\vspace{-24pt}
\centering
	   \includegraphics[width=.7in]{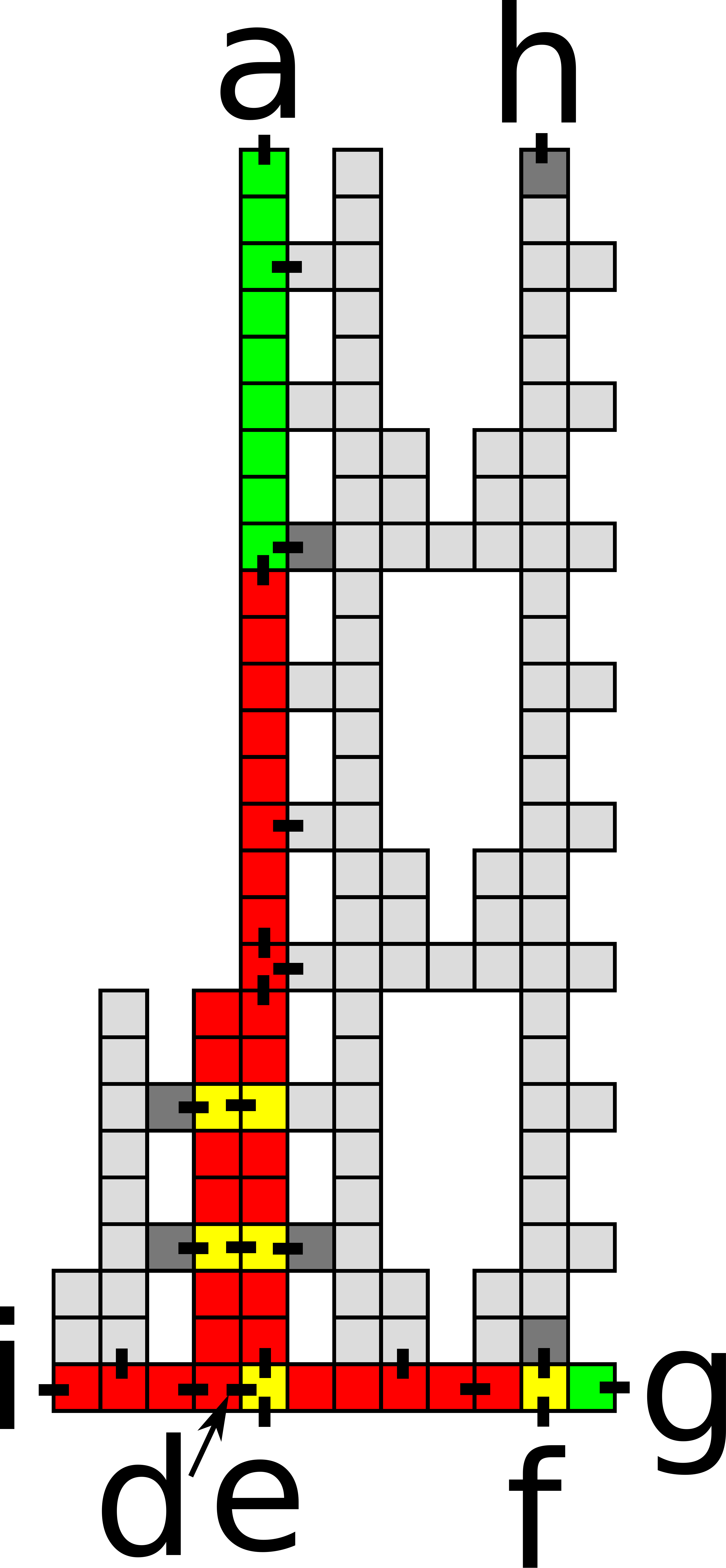}
   \caption{A schematic depiction of a grouted stage-$2$ ladder supertiles bound to a grouted right rung supertiles. \vspace{-26pt}} \label{fig:stage3rung}
\end{wrapfigure}
Then, just as glues $a$ and $h$ allow for a grouted stage-$2$ supertile to bind to glues of north edges of tiles of $L'$, $e$ and $f$ permit a grouted stage-$2$ supertile to bind to glues of south edges of tiles of $L'$. The glue labeled $g$ will either be an indicating glue or a generic glue (an $e$ glue in particular) depending on the type of grout that binds to $L$. If the grout type corresponds to type $r[ef]c$ or $r[ef]d$, then $g$ will be an indicating glue corresponding to the indicating glue of a tile of a stage-$2$ ladder supertile of type $r[ef]c$ or $r[ef]d$ respectively. The $d$ glue allows for grout supertiles to continue to bind after a grouted right rung supertiles binds. This scenario is depicted in Figure~\ref{fig:stage3rung}. Finally, the glue labeled $i$ in Figure~\ref{fig:stage3rung} encodes the grout type.

Now let $M$ be a stage-$2$ ladder supertile with type $r[ef]a$.  We refer to the glue labels for the glues shown in the leftmost figure in Figure~\ref{fig:grout-sequence1}. Most of these glues serve similar purposes to the glues of $L$ and there are two main differences. First, $a$ will either be a generic glue, $n$, or a glue which serves the same purpose as the glue $h$ in $L$. In the latter case, we call $a$ a ``helper glue''. If the type of grout that binds to $M$ is type $r[ef]b$ or $r[ef]c$, then $a$ will be a helper glue. This helper glue will facilitate the self-assembly of a stage-$4$ ladder supertile. If the type of grout that binds to $M$ is any other type of grout, then, $a$ is a generic glue. Finally, if the type of grout that binds to $M$ is $r[ab]a$, $r[ab]b$, $r[ab]c$, or $r[ab]d$, then the glue labeled $g$ is an indicating glue that is identical to the corresponding indicating glue of an edge of a tile in a stage-$2$ ladder supertile with type $r[ab]a$, $r[ab]b$, $r[ab]c$, or $r[ab]d$. Otherwise, $g$ will be a generic $e$ glue.

Next let $N$ be a stage-$2$ ladder supertile with type $r[ef]c$.  We refer to the glue labels for the glues shown in the third figure from the left in Figure~\ref{fig:grout-sequence1}. Once again, most of these glues serve similar purposes to the glues of $L$ or $M$. The main difference is that the $d$ glue is a generic $s$ glue and thus grout does not bind to the south edges of the southernmost tiles of $N$.
This is crucial for allowing grout to bind along these south edges in the assembly of higher ladder stages. At this point, we also note that there are two versions of stage-$2$ ladder supertile with type $r[ef]c$. The first version has two indicating glues, one on the east edge of each of the blue tiles in Figure~\ref{fig:grout-sequence1}, and the second version has generic $e$ glues instead of these indicating glues. Moreover, there are two versions of grout supertiles with type $r[ef]c$.  Grout with type $r[ab]a$, $r[ab]b$, $r[ab]c$ (both versions), or $r[ab]d$ can only bind to a stage-$2$ ladder with type $r[ef]c$ iff the type is of the first version. The purpose of the indicating glues on edges of these blue tiles will are utilized in the self-assembly of ladder supertiles of stage $\geq 4$

Finally let $P$ be a stage-$2$ ladder supertile with type $r[ef]d$.  We refer to the glue labels for the glues shown in the fourth figure from the left in Figure~\ref{fig:grout-sequence1}. Once again, most of these glues serve similar purposes to the glues of $N$. However, in this this case, there is one major difference. Namely, grout supertiles not only bind to the west edges of tiles of $P$, but they also bind to east edges as well. The green supertile with tiles containing edges with glues $g$ and $h$ initiate such growth. The glue labeled $h$ (resp. $e$) is a generic $n$ (resp. $s$) glue. The glues labeled $g$ and $f$ are binding glues. Glues $g$ and $h$ do not encode a grout type and are identical to the binding glues of a right rung supertile. This allows a grouted $P$ to serve the purpose of a grouted right rung supertile in the self-assembly of a stage-$4$ ladder.

\begin{wrapfigure}{r}{0.5\textwidth}
\vspace{-30pt}
\centering
	   \includegraphics[width=2in]{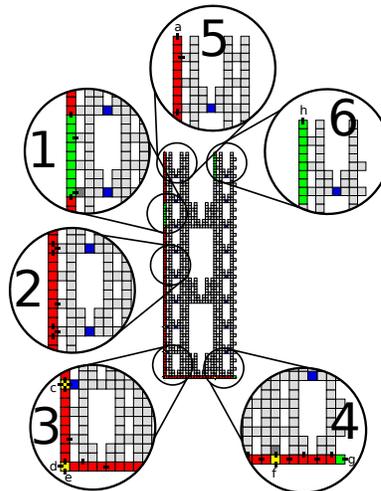}
   \caption{A schematic depiction of a grouted stage-$3$ supertile. Note the similarity between the pattern of glues labeled here and the glues of the second figure from the left in Figure~\ref{fig:grout-sequence1}. Many of the glues not depicted here are strength-$2$ glues which are hard-coded to allow either grout supertiles to self-assemble, stage-$2$ ladder supertiles to self-assemble, or rung supertiles to self-assemble. Glues depicted as strength-$1$ glues are intended to indicate how grout supertiles cooperatively bind. Glue labels shown here are for reference purposes only and are not the labels in the definition of the tile set for $\calT_{\bf{U}}$.\vspace{-16pt}} \label{fig:grout-sequence-main}
\end{wrapfigure}	

Note that tile types which self-assemble grout supertiles that bind to stage-$2$ ladder and rung supertiles can be defined so that
1) tiles at locations corresponding to yellow squares in Figure~\ref{fig:grouted-stage2} contain edges with binding glues that permit the self-assembly a stage-$3$ ladder supertile,
and 2) binding glues depend (though not necessarily all of the glues will) on the type of grout which binds.
Binding glues enable appropriate grouted stage-$2$ ladder and/or rung supertiles to bind to yield a stage-$3$ ladder supertile. We also note that tile types which self-assemble grout supertiles can be defined so that
1) the grouted stage-$2$ ladder and/or rung supertiles which bind to yield a stage-$3$ ladder supertile all contain the same type of grout supertiles,
2)  tiles at locations corresponding to green squares in Figure~\ref{fig:grouted-stage2} contain edges with indicating glues,
and 3) the indicating glues of an edge of a tile in a stage-$3$ ladder supertile are identical to the indicating glues of exactly one type of stage-$2$ ladder supertile; which type depends on the type of grout supertiles contained in the stage-$3$ ladder supertile.

\vspace{-15pt}
\subsubsection{Finite self-assembly of stage-$n$ ladder supertiles for $n\geq 2$}\label{sec:stage-n}

In Section~\ref{sec:grout-supertiles} we saw that tile types can be defined to self-assemble base supertiles and grout supertiles such that there is an assembly sequences where these supertiles bind to yield stage-$3$ ladder supertiles.
Moreover, the stage-$3$ ladder supertiles which self-assemble contain tiles with edges that contain indicating glues that are identical to the indicating glues to one of the stage-$2$ ladder supertile types, giving $10$ types of stage-$3$ ladder supertiles.

For $n\geq 3$, we note that copies of the same grout supertiles which bind to stage-$2$ ladder and rung supertiles can bind to stage-$(n-1)$ ladder supertiles, yielding grouted stage-$(n-1)$ supertiles such that appropriate grouted stage-$(n-1)$ supertiles can bind to yield a stage-$n$ ladder supertile. Moreover,  the stage-$n$ ladder supertiles which self-assemble contain tiles with edges that contain indicating glues that are identical to the indicating glues to one of the stage-$(n-1)$ ladder supertile types, and thus identical to indicating glues of one of the stage-$2$ ladder supertiles. See Figure~\ref{fig:grout-sequence-main} for a depictions of how grout supertiles bind to a stage-$3$ ladder supertile with type $r[ef]b$.

\vspace{-16pt}
\subsection{Final remarks}
\vspace{-4pt}

Theorems~\ref{thm:H-impossible} and~\ref{thm:U-impossible} show that the $H$-fractal and the $U$-fractal cannot be finitely self-assembled by any aTAM system. Therefore, Theorem~\ref{thm:U-fractal} shows the power that hierarchical self-assembly has over single tile attachment by showing that there is 2HAM system which finitely self-assembles the $U$-fractal. We conjecture that one can also give a 2HAM system that finitely self-assembles the $H$-fractal.%

\vspace{-8pt}
\begin{conjecture}\label{conj:H-fractal}
Let $\bf{H}$ be the H-fractal DSSF. There exists a 2HAM TAS $\calT_{\bf{H}} = (T_{\bf{H}}, 2)$ that finitely self-assembles $\bf{H}$.
\end{conjecture}
\vspace{-8pt}

We've described the self-assembly of stage-$n$ ladder supertiles via particular assembly sequences of $\calT_{\bf{U}}$, ignoring many others and many producible supertiles. 
Section~\ref{sec:construction-appendix} 
describes how our construction ensures finite self-assembly of $\bf{U}$ despite these many possible assembly sequence and producible supertiles. Finally, our system self-assembles higher and higher stages of the ladder supertiles.
Note that $\bf{U}$, by definition, only contains points in the first quadrant of the plane. Moreover, the westernmost points (resp. southernmost points) are a vertical (resp. horizontal) line of points. We call these points the ``boundary'' of $\bf{U}$. Only self-assembling higher and higher stages of ladder supertiles would give a system that finitely self-assembles $\bf{U}$ without points on the boundary.  
In Section~\ref{sec:construction-appendix} 
we give a simple tweak that ensures there is an assembly sequence from any producible assembly sequence to a terminal assembly with domain equal to $\bf{U}$ (including boundary points).
\vspace{-6pt}

\vspace{-10pt}
\bibliographystyle{abbrv} %
\bibliography{tam,experimental_refs}

\newpage
\appendix
\begin{center}
	\Huge\bfseries
	Technical Appendix
\end{center}

\section{Informal definition of the aTAM}\label{sec:atam-informal}

This section gives a brief definition of the abstract Tile Assembly Model (aTAM) and related terminology. We use definitions, terminology, and notation from \cite{SSADST}.
A \emph{tile type} is a unit square with four sides, each consisting of a \emph{glue label}, often represented as a finite string, and a nonnegative integer \emph{strength}. A glue~$g$ that appears on multiple tiles (or sides) always has the same strength~$s_g$. %
There are a finite set $T$ of tile types, but an infinite number of copies of each tile type, with each copy being referred to as a \emph{tile}. An \emph{assembly}
is a positioning of tiles on the integer lattice $\Z^2$, described  formally as a partial function $\alpha:\Z^2 \dashrightarrow T$. %
Let $\mathcal{A}^T$ denote the set of all assemblies of tiles from $T$, and let $\mathcal{A}^T_{< \infty}$ denote the set of finite assemblies of tiles from $T$.
We write $\alpha \sqsubseteq \beta$ to denote that $\alpha$ is a \emph{subassembly} of $\beta$, which means that $\dom\alpha \subseteq \dom\beta$ and $\alpha(p)=\beta(p)$ for all points $p\in\dom\alpha$.
Two adjacent tiles in an assembly \emph{interact}, or are \emph{attached}, if the glue labels on their abutting sides are equal and have positive strength. %
Each assembly induces a \emph{binding graph}, a grid graph whose vertices are tiles, with an edge between two tiles if they interact.
The assembly is \emph{$\tau$-stable} if every cut of its binding graph has strength at least~$\tau$, where the strength   of a cut is the sum of all of the individual glue strengths in the cut.

A \emph{tile assembly system} (TAS) is a triple $\calT = (T,\sigma,\tau)$, where $T$ is a finite set of tile types, $\sigma:\Z^2 \dashrightarrow T$ is a finite, $\tau$-stable \emph{seed assembly},
and $\tau$ is the \emph{temperature} (i.e. the minimum binding threshold for a tile).
An assembly $\alpha$ is \emph{producible} if either $\alpha = \sigma$ or if $\beta$ is a producible assembly and $\alpha$ can be obtained from $\beta$ by the stable binding of a single tile.
In this case we write $\beta\to_1^\calT \alpha$ (to mean~$\alpha$ is producible from $\beta$ by the attachment of one tile), and we write $\beta\to^\calT \alpha$ if $\beta \to_1^{\calT*} \alpha$ (to mean $\alpha$ is producible from $\beta$ by the attachment of zero or more tiles).
When $\calT$ is clear from context, we may write $\to_1$ and $\to$ instead.
An assembly sequence in a TAS $\mathcal{T}$ is a (finite or infinite) sequence $\vec{\alpha} = (\alpha_0,\alpha_1,\ldots)$ of assemblies in which each $\alpha_{i+1}$ is obtained from $\alpha_i$ by the addition of a single tile. The \emph{result} $\res{\vec{\alpha}}$ of such an assembly sequence is its unique limiting assembly. (This is the last assembly in the sequence if the sequence is finite.)
Let $\vec{\alpha}$ be an assembly sequence. In the following, $\vec{\alpha}[i]$ denotes the tile that $\vec{\alpha}$ places at assembly step $i$. We say that $\vec{\alpha}[i]$ is the \emph{parent} of $\vec{\alpha}[j]$ if $i < j$ and $\vec{\alpha}[j]$ binds to $\vec{\alpha}[i]$. Furthermore, we say that tile $\vec{\alpha}[i]$ is the \emph{ancestor} of a tile $\vec{\alpha}[k]$ if either $\vec{\alpha}[i]$ is the parent of $\vec{\alpha}[k]$, or there exists an index $j$, such that, $i < j < k$, $\vec{\alpha}[j]$ is the parent of $\vec{\alpha}[k]$ and $\vec{\alpha}[i]$ is the ancestor of $\vec{\alpha}[j]$. Note that $\vec{\alpha}[j]$ implicitly refers to both the type of tile and its location, and the parent and ancestor relationships, in general, depend on the given assembly sequence $\vec{\alpha}$.

We let $\prodasm{\calT}$ denote the set of producible assemblies of $\calT$.
An assembly is \emph{terminal} if no tile can be $\tau$-stably attached to it.
We let   $\termasm{\calT} \subseteq \prodasm{\calT}$ denote  the set of producible, terminal assemblies of $\calT$.
A TAS $\calT$ is \emph{directed} if $|\termasm{\calT}| = 1$. %
A set $X$ \emph{strictly self-assembles} in a TAS $\calT$ if
every assembly $\alpha\in\termasm{T}$ satisfies $\dom \alpha =
X$. Essentially, strict self-assembly means that tiles are placed
in exactly the positions defined by the shape of $X$, and this is true for all assembly sequences of $\calT$.
For an infinite shape $X\subseteq \Z^2$, we say that $\calT$ finitely self-assembles $X$ if every finite producible assembly of $\calT$ has a possible way of growing into an assembly that places tiles exactly on those points in $X$.

\section{Informal definition of the 2HAM}\label{sec:2ham-informal}

The 2HAM \cite{AGKS05g,DDFIRSS07} is a generalization of the aTAM in that it allows for two assemblies, both possibly consisting of more than one tile, to attach to each other. Since we must allow that the assemblies might require translation before they can bind, we define a \emph{supertile} to be the set of all translations of a $\tau$-stable assembly, and speak of the attachment of supertiles to each other, modeling that the assemblies attach, if possible, after appropriate translation. We use definitions, terminology, and notation from \cite{Versus}.
We now give a brief, informal, sketch of the 2HAM.

A \emph{tile type} is a unit square with each side having a \emph{glue} consisting of a \emph{label} (a finite string) and \emph{strength} (a non-negative integer).   We assume a finite set $T$ of tile types, but an infinite number of copies of each tile type, each copy referred to as a \emph{tile}.
A \emph{supertile} is (the set of all translations of) a positioning of tiles on the integer lattice $\Z^2$.  Two adjacent tiles in a supertile \emph{interact} if the glues on their abutting sides are equal and have positive strength.
Each supertile induces a \emph{binding graph}, a grid graph whose vertices are tiles, with an edge between two tiles if they interact.
The supertile is \emph{$\tau$-stable} if every cut of its binding graph has strength at least $\tau$, where the weight of an edge is the strength of the glue it represents.
That is, the supertile is stable if at least energy $\tau$ is required to separate the supertile into two parts.  Note that throughout this paper, we will use the term \emph{assembly} interchangeably with supertile.

A \emph{(two-handed) tile assembly system} (\emph{TAS}) is an ordered triple $\mathcal{T} = (T, S, \tau)$, where $T$ is a finite set of tile types, $S$ is the \emph{initial state}, and $\tau\in\N$ is the temperature.  For notational convenience we sometimes describe $S$ as a set of supertiles, in which case we actually mean that $S$ is a multiset of supertiles with one count of each supertile. We also assume that, in general, unless stated otherwise, the count for any single tile in the initial state is infinite.  Commonly, 2HAM systems are defined as pairs $\mathcal{T} = (T, \tau)$, with the initial state simply consisting of an infinite number of copies of each singleton tile type of $T$, and throughout this paper this is the notation we will use.

Given a TAS $\calT=(T,\tau)$, a supertile is \emph{producible}, written as $\alpha \in \prodasm{T}$, if either it is a single tile from $T$, or it is the $\tau$-stable result of translating two producible assemblies without overlap.
A supertile $\alpha$ is \emph{terminal}, written as $\alpha \in \termasm{T}$, if for every producible supertile $\beta$, $\alpha$ and $\beta$ cannot be $\tau$-stably attached. Informally, an \emph{assembly sequence} of a TAS $\calT$ is a sequence of states (sets of supertiles) $\vec{S} = (S_i \mid 0 \leq i < k)$ (where $k = \infty$ if $\vec{S}$ is an infinite assembly sequence) where $S_{i+1}$ is obtained from $S_i$ by picking two supertiles from $S_i$ that can attach to each other, and attaching them. We always assume that $S_0$ is contrained so that an assembly sequence tends toward a unique supertile. The \emph{result} of an assembly sequence is then defined to be this unique supertile.

A TAS is \emph{directed} if it has only one terminal, producible supertile. A set, or shape, $X$ \emph{strictly self-assembles} if there is a TAS $\mathcal{T}$ for which every assembly $\alpha\in\termasm{T}$ satisfies $\dom \alpha = X$. Essentially, strict self-assembly means that tiles are only placed in positions defined by the shape.  This is in contrast to the notion of \emph{weak self-assembly} in which only specially marked tiles can and must be in the locations of $X$ but other locations can perhaps receive tiles of other types.
For an infinite shape $X\subseteq \Z^2$, we say that $\calT$ finitely self-assembles $X$ if every finite producible assembly of $\calT$ has a possible way of growing into an assembly that places tiles exactly on those points in $X$. In this paper we consider finite self-assembly of DSSF's (in the strict sense).

\section{Discrete Self-Similar Fractals}\label{sec:dssf-def}

We define $\mathbb{N}_g$ as the subset $\{0,1,...,g-1\}$ of $\mathbb{N}$, and
if $A,B \subseteq \mathbb{N}^2$, then $A+(x,y)B = \{(x_a,y_a) + (x\cdot x_b,y \cdot y_b) | (x_a,y_a) \in A$ and $(x_b,y_b) \in B\}$.  We then define discrete self-similar fractals as follows:

We say that $\bf{X} \subset \mathbb{N}^2$ is a \emph{discrete self-similar fractal} (or \emph{DSSF} for short) if there
exists a set $\{(0,0)\} \subset G \subset \mathbb{N}^2$ where $G$ is connected, $w_G = \max(\{x|(x,y) \in G\})+1$, $h_G = \max(\{y|(x,y)\in G\})+1$, $w_G$ and $h_G > 1$, and $G \subsetneq \mathbb{N}_{w_G} \times \mathbb{N}_{h_G}$, such that $\bf{X}$ $ = \bigcup^\infty_{i=1} X_i$, where $X_i$, the $i^{th}$ stage of $\bf{X}$, is defined by $X_1 = G$ and $X_{i+1} = X_i + (w_G^i,h_G^i)G$.  We say that $G$ is the generator of $\bf{X}$.  Essentially, the generator is a connected set of points in $\mathbb{N}^2$ containing $(0,0)$, points at both $x>0$ and $y>0$, and is not a completely filled rectangle.  Every stage after the constructor is composed of copies of the previous stage arranged in the same pattern as the generator.

A connected discrete self-similar fractal is one in which every component is connected in every stage, i.e. there is only one connected component in the grid graph formed by the points of the shape.

Figure~\ref{fig:triangle} shows, as an example, the first $4$ stages of the discrete self-similar fractal known as the Sierpinski triangle.  In this example, $G = \{(0,0),(1,0),(0,1)\}$.

\begin{figure}[htp]
\begin{center}
\vspace{-15pt}
\includegraphics[width=2.0in]{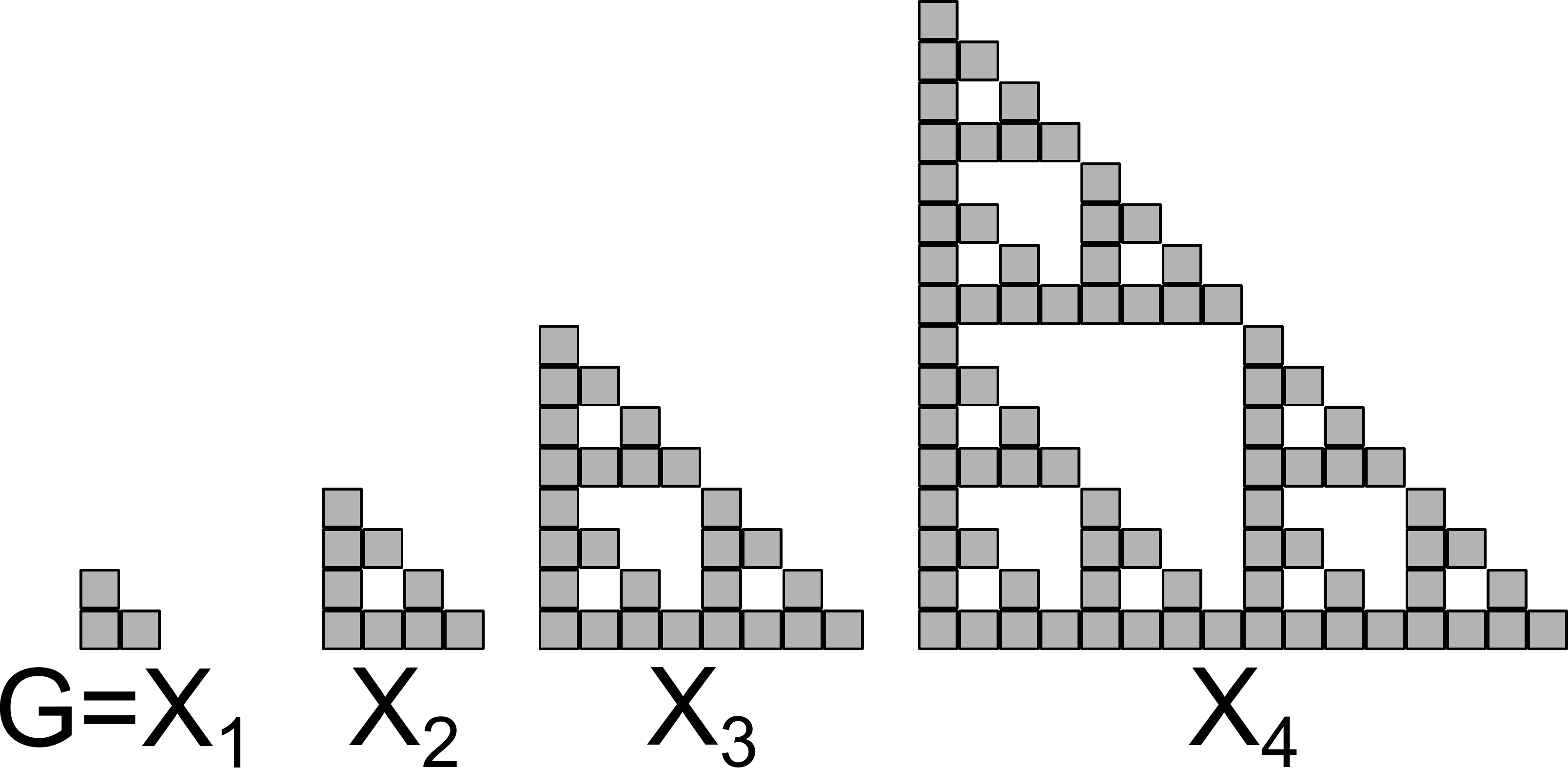}
\caption{Example discrete self-similar fractal:  the first $4$ stages of the Sierpinski triangle\vspace{-10pt}}
\label{fig:triangle}
\end{center}
\end{figure}

\section{Full proof of the impossibility of finite self-assembly of the $\Hfrac$ fractal in the aTAM}\label{sec:aTAM-full-impossible}

In this section, we give the full proof of Theorem~\ref{thm:H-impossible} showing that $\Hfrac$ does not finitely self-assemble in the aTAM.

\begin{proof}
For the sake of obtaining a contradiction, assume there exists an aTAM TAS $\calT = (T,\sigma,\tau)$ in which $\Hfrac$ finitely self-assembles.  We will show that $\Hfrac$ does not finitely self-assemble in $\mathcal{T}$.  Without loss of generality, we will assume that $|\sigma| = 1$, i.e. that $\calT$ is singly-seeded but our proof technique will hold for any TAS $\mathcal{T}$ with finite seed assembly. Since the location of $\sigma$ must be within $\Hfrac$, let $s$ be the stage number of the smallest stage of $\Hfrac$ which contains $\sigma$.

Let $c = 6|T|^6$.  If $\Hfrac$ finitely self-assembles in $\calT$, then every producible assembly in $\calT$ has domain contained in $\Hfrac$. Let $\vec{\alpha}$ be the shortest assembly sequence in $\calT$ whose result has domain $h_{c+s+2}$ (where $h_i$ is the $i$th stage of $\Hfrac$), subject to the additional constraint that, when multiple locations could receive a tile in a given step, $\vec{\alpha}$ always places a tile in a location of the smallest possible stage.

By our choice of $c$, we know that there are at least 6 stages of $\Hfrac$ whose respective bottleneck points are identically tiled by $\vec{\alpha}$.
Since, in any assembly sequence, the center tile of each stage of $\Hfrac$ either has a parent adjacent to the left or right, it follows, without loss of generality, that there are at least 3 stages, namely $h_i$, $h_j$ and $h_k$, for $i < j < k$, whose respective bottleneck points are identically tiled by $\vec{\alpha}$ and whose, respective center tiles have parents adjacent to the left.

\begin{figure}[htp]
\begin{center}
\includegraphics[width=2.0in]{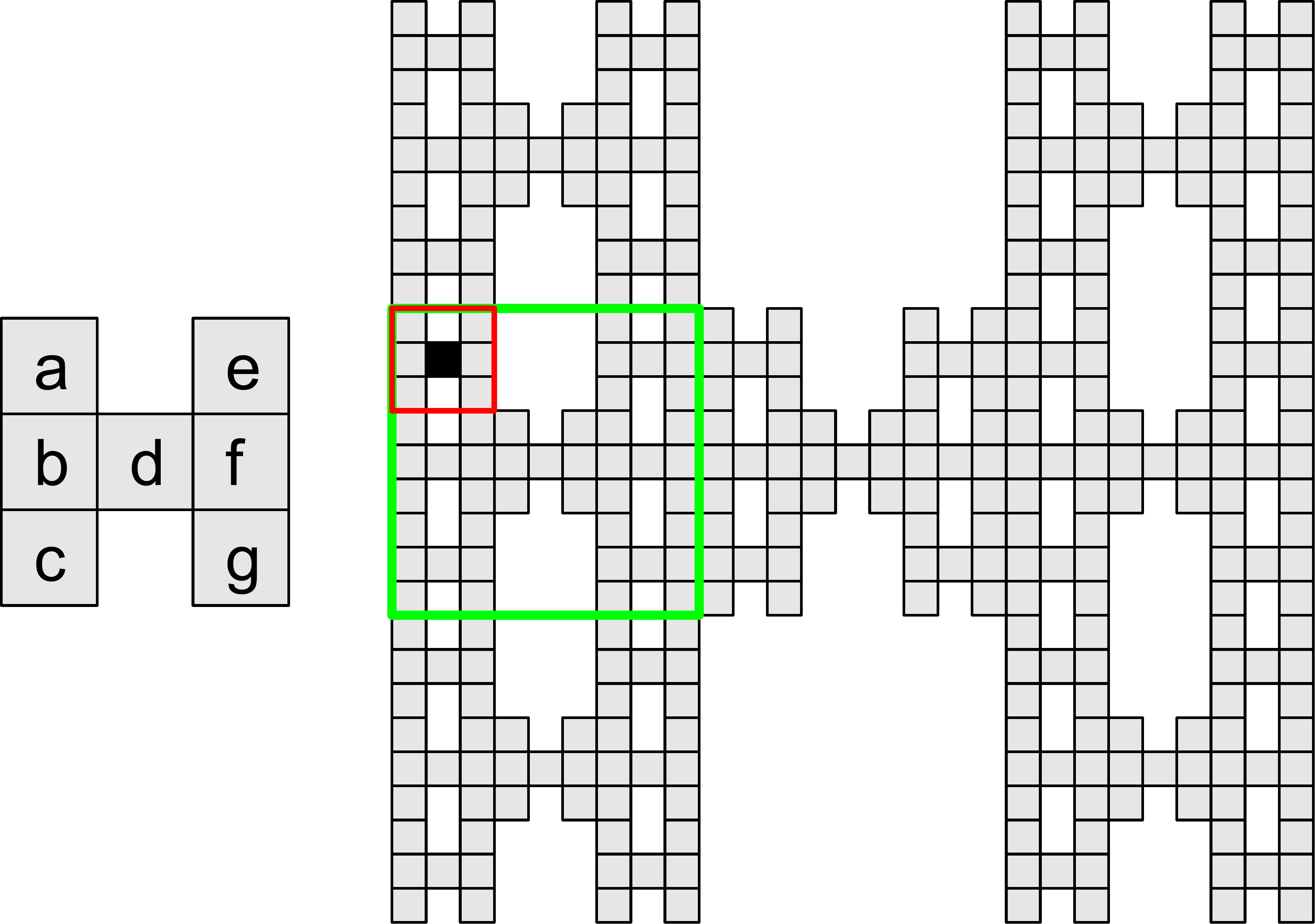}
\caption{(left) Address labels of each point in the generator of $\Hfrac$, (right) The black location is contained within stage three, and its address is $dab$ (i.e. it is location $d$ in a stage one copy (outlined in red), within location $a$ of a stage two copy (outlined in green), within location $b$ of stage three.)}
\label{fig:H-addresses}
\end{center}
\end{figure}

Relative to $\vec{\alpha}$, there are three cases to consider: (1) (2) some top-left-placed (bottom-left-placed) tile of the left-center of $h_j$ is placed at a point that is vertical distance $1/3$ toward the middle-left bottleneck tile of $h_j$ or (3) middle-left-placed tiles in $h_j$ are placed at points at least vertical distance $2/3$ up \emph{and} down toward the top-left and bottom-left bottleneck tiles, respectively. Note that, if none of these cases apply, then the left-center of $h_j$ would not assemble completely and $\Hfrac$ would not finitely self-assemble in $\mathcal{T}$.

\paragraph {\bf Case 1:} First, assume that some top-left-placed tile of the left-center of $h_j$ is placed at a point vertical distance $1/3$ toward the middle-left bottleneck tile of $h_j$.

We assign to each point $\vec{x} \in \Hfrac$ an address, which specifies its relative location within $\Hfrac$ (see Figure~\ref{fig:H-addresses}).

By the definition of $\Hfrac$, we can see that, for stage $j > 1$, the address of the top-left bottleneck tile is $d^{j-2}ad$ and its $y$-coordinate is $3^{j-1}+\left(3^{j-2}-1\right)/2 + 2\cdot 3^{j-2}$. However, the $y$-coordinate of the middle-left bottleneck tile is $3^{j-1}+\left(3^{j-2}-1\right)/2 + 3^{j-2}$.  Therefore, the vertical distance between those two is $3^{j-2}$, and a tile $1/3$ of the way down to the middle-left bottleneck tile has a $y$-coordinate $3^{j-2}/3 = 3^{j-3}$, which is less than that of the top-left bottleneck tile of $h_j$.

Given that the address of the top-left bottleneck tile is $d^{j-2}ad$, we know that it is located in the center (i.e. point $d$ in the generator) position of $j-2$ (appropriately-translated) sub-stages.  Its height above the bottom of the second-to-last of those initial $d$ positions is $\left(3^{j-3}-1\right)/2$, which is less than $3^{j-3}$. This means that a top-left-placed tile located $1/3$ of the way down toward the middle-left bottleneck tile must be located at a point that is lower than the bottom of the bounding box containing the sub-stage corresponding to the second-to-last of the initial $d$ positions. This also means that a sequence of tile placements to it must travel beyond (to the right of) the boundaries of that sub-stage. An example of this can be seen by the black line in Figure~\ref{fig:H5-top-left-bottleneck-zoomed}, which is outside of the blue box.

We are now in a position to create a new valid assembly sequence in $\mathcal{T}$ as follows.  Starting from the seed, run $\vec{\alpha}$ until the step at which it places the first bottleneck tile on the left side of $h_j$. Then, begin recording a sub-sequence of $\vec{\alpha}$ and denote this sub-sequence as $\vec{\alpha}'$.  As we run $\vec{\alpha}$ forward from this point, until it places the last tile of $h_j$, whenever a top-left-placed tile in $h_j$ is placed by $\vec{\alpha}$, we add that tile placement (type and location) to $\vec{\alpha}'$.  In this way, $\vec{\alpha}'$ becomes a sub-sequence of $\vec{\alpha}$ that records the growth of the top-placed sub-assembly -- and only the top-placed sub-assembly -- of the left-center of $h_j$.

\begin{figure}[htp]
\begin{center}
\includegraphics[width=2.0in]{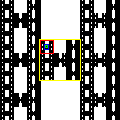}
\caption{Example showing the top-left bottleneck location of the 5-th stage of $\Hfrac$.  The address of this bottleneck tile is $d^3ad$, and the first sub-stage for which it is in location $d$ is surrounded by a green box, the second by a blue, the third by a red. Then, the sub-stage for which the location of that sub-stage is in location $a$ is surrounded by a yellow box.  This, then, is in position $d$ of stage 5.}
\label{fig:H5-top-left-bottleneck}
\end{center}
\end{figure}

\begin{figure}[htp]
\begin{center}
\includegraphics[width=4.0in]{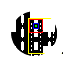}
\caption{Zoomed in portion of the example from Figure~\ref{fig:H5-top-left-bottleneck}.  The black line marks the distance $1/3$ of the way down toward the middle-left bottleneck.  Growth from the top-left bottleneck tile to any location adjacent to this line must exceed the bounding boxes of the first $2$ sub-stages, and must do so by first growing directly to the right of them.}
\label{fig:H5-top-left-bottleneck-zoomed}
\end{center}
\end{figure}

At this point, we are ready to show that $\Hfrac$ does not finitely self-assemble in $\mathcal{T}$. First, reset $\vec{\alpha}$ to the seed and begin its forward growth until the placement of the first bottleneck tile on the left side of $h_i$.  At this point, we will merge $\vec{\alpha}$ and $\vec{\alpha}'$ as follows.  For each tile position $\vec{p}$ in $\vec{\alpha}'$, we translate it so that the new position, $\vec{p}'$, is the point with the same relative offset from the top-left bottleneck position of $h_i$ as $\vec{p}$ was from the top-left bottleneck position of $h_j$. Namely, we have $\vec{p}' = \vec{p} - \left(3^{j-1}+\frac{3^{j-2}-1}{2},3^{j-1}+\frac{3^{j-2}-1}{2}\right) + 3^{j-2}(0,2) + \left(3^{i-1}+\frac{3^{i-2}-1}{2},3^{i-1}+\frac{3^{i-2}-1}{2}\right) + 3^{i-2}(0,2)$).  We continue to run $\vec{\alpha}$ forward by performing all tile placements up to, and including, the placement of the top-left bottleneck tile of $h_i$, with the exception of the middle or bottom-left bottleneck tiles or any tiles whose ancestors, relative to $\vec{\alpha}$, are either the middle or bottom-left bottleneck tiles. As soon as $\vec{\alpha}$ places the top-left bottleneck tile in $h_i$, we follow the tile placements of the modified $\vec{\alpha}'$. The result is a valid assembly sequence up to the point of the placement of at least one tile outside of $\Hfrac$. Thus, $\Hfrac$ does not finitely self-assemble in $\calT$.

\begin{figure}[htp]
\begin{center}
\includegraphics[width=3.0in]{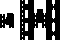}
\caption{Depiction of how top-placed growth from stage 5 would go out of bounds of $\Hfrac$ in stage 3 (left) and stage 4 (right).  The black tile is the top-left bottleneck tile, the green locations are those which correctly match the smaller stage, and the red are those which go out of bounds of $\Hfrac$. Clearly, all tiles in green positions will be able to grow, and then erroneous growth is forced to occur immediately to the right of the green tiles, where no other tiles could prevent this growth.}
\label{fig:H5-bad-growth-in-3-and-4}
\end{center}
\end{figure}

\paragraph {\bf Case 2:}  This case is symmetric to the previous case and is therefore omitted.

\paragraph {\bf Case 3:}  Here, middle-placed tiles of $h_j$ are placed by $\vec{\alpha}$ at points that are $2/3$ of the way up \emph{and} down toward the top-left \emph{and} bottom-left bottleneck tiles.
We focus on the upward growth and show that this will result in an erroneous tile placement.

First, we note that the address of the middle-left bottleneck tile of $h_j$ is $d^{j-2}bd$.  Then, we see that points within the left-center of stage $h_j$, which are vertical distance $2/3$ toward the top-left bottleneck tile, are beyond the top-most boundary of the bounding box around the initial $d^{j-2}$ sub-stages containing the middle-left bottleneck. Specifically, any sequence of tile placements by $\vec{\alpha}$ to such a point must pass through the top-right-most side of that bounding box.
We will now create an assembly sub-sequence $\vec{\alpha}''$ that records the tile placements of only the middle-placed tiles of $h_j$, similar to the definition of $\vec{\alpha}'$ in {\bf Case 1}. Then, we run $\vec{\alpha}$ forward, starting from the seed, performing all tile placements up to, and including, the placement of the middle-left bottleneck tile of $h_k$, with the exception of the top or bottom-left bottleneck tiles or any tiles whose ancestors, relative to $\vec{\alpha}$, are either the top-left or bottom-left bottleneck tiles. As soon as $\vec{\alpha}$ places the middle-left bottleneck tile in $h_k$, we follow the tile placements of the modified $\vec{\alpha}''$. The result is a valid assembly sequence.

\begin{figure}[htp]
\begin{center}
\includegraphics[width=2.0in]{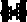}
\caption{A portion of the assembly of stage 6 of $\Hfrac$ depicting how middle-placed growth from stages 3, 4, or 5 would go out of bounds of $\Hfrac$ in stage 6.  The black tile is the middle-left bottleneck tile. The red tile shows a location that would grow out of bounds of $\Hfrac$ if the middle-placed tiles from stage 3 were allowed to grow in stage 6. The green tiles show the same for the middle-placed tiles from stage 4, and the blue show those from stage 5.}
\label{fig:H-middle-bad}
\end{center}
\end{figure}

It is worthy to note that the middle-left bottleneck of $h_k$ is located at address $d^{k-2}bd$. However, since $j < k$, the middle-left bottleneck tile of $h_k$ is nested deeper within location $d$ of (appropriately-translated) sub-stages of $\Hfrac$ than the middle-left bottleneck tile of stage $h_j$.  We now execute the sequence of attachments specified by $\vec{\alpha}''$, appropriately-translated from $h_j$ to $h_k$. By the difference of addresses of the respective middle-left bottleneck tiles of stages $h_j$ and $h_k$, $\vec{\alpha}''$ will place a tile (from $h_j$) up and out of the domain of $\Hfrac$ (within $h_k$). Thus, $\Hfrac$ does not finitely self-assemble in $\mathcal{T}$.
\qed

\end{proof}

\section{Details of the Impossibility of Finite Self-Assembly of the $\Ufrac$ fractal in the aTAM}\label{sec:u-appendix}

This section contains additional details for the proof of Theorem~\ref{thm:U-impossible}.

For notational purposes, we will refer to $u_i$ as the $i$th stage of $\Ufrac$.%
We call the \emph{center} tile of $u_i$ the single, unique central tile in the middle of the bottom row of the $i$th stage which connects the \emph{left} and \emph{right} sides of $u_i$.

\begin{definition}
Let $B^{\Ufrac}_0 = \{(0,0),(0,1),(0,2),(2,0),(2,1),(2,2)\}$.  For stages $i > 1$, we call the following set of 6 points the \emph{bottleneck} points of stage $i$, or $B^{\Ufrac}_i$:

$B^{\Ufrac}_i = \{(3^{i-1}+\frac{3^{i-2}-1}{2},0) + 3^{i-2}b | b \in B^{\Ufrac}_0 \}$.
\end{definition}

An example of the bottleneck points for a few stages of $\Ufrac$ can be seen in Figure~\ref{fig:U}.

\begin{proof}
We prove Theorem~\ref{thm:U-impossible} by contradiction.  Therefore, assume the converse, namely that there does exist an aTAM system which finitely self-assembles $\Ufrac$.  Let $\calT = (T,\sigma,\tau)$ be such a system.  We will now show that $\calT$ cannot correctly finitely self-assemble $\Ufrac$.  Without loss of generality, we will assume that $|\sigma| = 1$, i.e. that $\calT$ is singly-seeded.  Our proof, however, holds for any finite sized seed with trivial adjustments of a few constants. Since the location of $\sigma$ must be within $\Ufrac$, let $s$ be the stage number of the smallest stage of $\Ufrac$ which contains $\sigma$.

We define constant $c \in \mathbb{N}$ such that $c = 6(|T|^6)$.  Now, assuming that $\calT$ correctly finitely self-assembles $\Ufrac$, then every producible assembly in $\calT$ has domain contained in $\Ufrac$.  Therefore, we will now allow assembly to begin from the seed, $\sigma$, and proceed until the final tile of stage $h_{c+s+2}$ attaches (i.e. we allow it to grow $c$ stages past the smallest stage containing the seed plus 2).  Throughout the assembly process, we bias the assembly sequence so that whenever multiple locations could receive tiles during any given step, it always chooses to add a tile in a location of the smallest possible stage.  Note that as it still only chooses from valid tile attachments at every step, this results in a valid assembly sequence of $\calT$, which we record and refer to as $\vec{\alpha}$.

We now inspect $\vec{\alpha}$ and sort each stage $u_i$ for $2 < i \leq c$ into groups of sets.  By the definition of $c$, since each stage has $6$ bottleneck locations which each may have any of $|T|$ tile types (note that since $h_{c+s+2}$ has completed at this point, so have all its substages by definition of $\Ufrac$) for a total of $|T|^6$ ways to tile the bottleneck locations of a stage, it must be the case that there is a set of at least $6$ stages which have the exact same tile types in each of their respective bottleneck locations.  We refer to this set of stages as $S$.

We now inspect $\vec{\alpha}$ and $S$ and sort each stage $u_i \in S$ into set $S_L$ or $S_R$ based on whether or not, in $\vec{\alpha}$, a tile was first placed immediately to the left of $u_i$'s center tile (putting it in $S_L$), or a tile was first placed immediately right of $u_i$'s center tile (putting it in $S_R$).  Because the center tile of $u_i$ must eventually be placed, at least one of those locations must first receive a tile since a tile placed at at least one of those locations must be a parent of the center tile.   Note that since $|S| = 6$ and there are only two choices, either $|S_L|$ or $|S_R|$ must be at least $3$.  Without loss of generality, we'll assume $S_L$ is at least this large, but the rest of the argument is identical but symmetric if only $S_R$ is this large.

We will refer to the portion of stage $u_i$ which lies to the east of $u_i$'s left side bottleneck tiles, but west of its center tile, as the \emph{left-middle} of $u_i$.  It must be the case that in $\vec{\alpha}$, during the growth of the left-middle of $u_i$, each of the tiles there had as an ancestor at least one of the bottleneck tiles.  For each tile in the left-middle of $u_i$, we will call it \emph{top-placed} iff it has as an ancestor a single bottleneck tile of stage $u_i$ and that bottleneck tile is the top (left) bottleneck tile.  Analogously, we denote \emph{middle-placed} and \emph{bottom-placed} tiles if their sole bottleneck tile ancestors are the middle and bottom (left) bottleneck tiles, respectively.  Note that some tiles may fall into none of these categories, as they may have multiple bottleneck tiles as ancestors, but that all tiles in the left-middle of stage $u_i$ must have at least one of the left bottleneck tiles of stage $u_i$ as ancestors.  This is because the only paths of growth from the seed to tiles in those locations are through the bottleneck tiles because we know that in $u_i$, the tile immediately to the left of the center tile is placed before the tile immediately to its right, meaning that growth could not have proceeded from right to left through the center tile, and thus the only possible paths of growth into the left-middle of $u_i$ are through the left bottleneck tiles.  We now point out the fact that in order for there to be any tiles of the left-middle of $u_i$ not in those categories, they must have as ancestors multiple unique left bottleneck tiles of $u_i$, and this is only possible if some tile in their ancestry (or they themselves) were placed via cooperation between two or more tiles which have as ancestors each of those bottleneck tiles (or were themselves the bottleneck tiles).  Furthermore, this is only possible if tiles with unique bottleneck tiles in their ancestry grew to locations which were no further than one position separated, which is the only way a mixed-ancestry tile could be placed.  (A tile could also bind between a bottleneck tile itself and another whose ancestor is some other bottleneck tile, but the same argument holds.)  This in turn means that all tiles in the left-middle of $u_i$ are top-placed, left-placed, or bottom-placed, or tiles which were in two or more of those categories grew to locations within no more than one unit square away from each other.  The key point from this observation is that assemblies of tiles which are in those categories, meaning that in $u_i$ they could be grown simply from their associated bottleneck tiles without cooperation from any other tiles, either completely fill the space of the left-middle of $u_i$, or tiles of those categories grow vertically together to within at least a single space of each other.

\begin{figure}[htp]
\begin{center}
\includegraphics[width=2.0in]{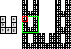}
\caption{(left) Address labels of each point in the generator of $\Ufrac$, (right) The black location is contained within stage three, and its address is $dab$ (i.e. it is location $d$ in a stage one copy (outlined in red), within location $a$ of a stage two copy (outlined in green), within location $b$ of stage three.)}
\label{fig:U-addresses}
\end{center}
\end{figure}

We now inspect the different possibilities related to how tiles grow throughout the left-middle of the stages contained in $S_L$.   These possibilities lead to several cases, at least one of which must occur, but we will show that any of them would violate the fact that $\calT$ finitely self-assembles $\Ufrac$.  Throughout the following, let $i < j < k$ be the indices of the three stages in $S_L$ (or, if there are more stages in $S_L$, any arbitrary subset of three of them such that $i < j < k$).  Note that we assign addresses to the locations in $\Ufrac$ as shown in Figure~\ref{fig:U-addresses}.

\begin{figure}[htp]
\begin{center}
\includegraphics[width=2.0in]{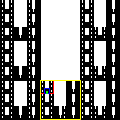}
\caption{Example showing the top left bottleneck location of the 5th stage of $\Ufrac$.  The address of this bottleneck tile is $d^3ad$, and the first substage for which it is in location $d$ is surrounded by a green box, the second by a blue, the third by a red, then the substage for which the location of that substage is in location $a$ is surrounded by a yellow box.  This, then, is in position $d$ of stage $5$.}
\label{fig:U5-top-left-bottleneck}
\end{center}
\end{figure}

\paragraph {\bf Case1: }   By observation of the shape of $\Ufrac$, we can see that for stage $j > 1$, the address of the top left bottleneck tile is $d^{j-2}ad$ (see Figure~\ref{fig:U5-top-left-bottleneck}).  Assume that some top-placed tile of the left-middle of $u_j$ attaches in a location which is east of the level $j-3$ sub-stage in which it is contained.  This means that a path of tiles can grow solely from the top left bottleneck to such a location.  But, since a tile of the same type is also located in the top left bottleneck of stage $i < j$, this growth could also occur in stage $i$ and would go outside of the boundaries of $\Ufrac$.  (The technique for showing how this must be possible is the same as used in the proof of Theorem~\ref{thm:H-impossible}, and an example can be seen in Figure~\ref{fig:U5-bad-growth-in-3-and-4}.)
Thus, it is impossible for a top-placed tile of $u_j$ to grow that far to the east.  By inspecting the shape of $\Ufrac$, we see that the only connections to top-placed tiles that the substage containing the top left bottleneck has between it and the portion of the left-middle of $u_j$ below is through the south side of the level $j-2$ substage in which it is located (e.g. if $j=6$, then the address is $ddddad$, and the only connections to the south are through the level $4$ substage).  This means that in stage $u_j$, in order for a top-placed tile to attach in a position which is lower than the top-left bottleneck tile, it must be part of a path of tiles which have grown to the east outside of the level $j-3$ substage, which we showed is impossible.  Thus, we know that in stage $u_j$, no top-placed tile appears in a location lower than the top left bottleneck.

\begin{figure}[htp]
\begin{center}
\includegraphics[width=4.0in]{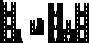}
\caption{Depiction of how top-placed growth from stage 5 would go out of bounds of $\Ufrac$ in stage 3 and stage 4.  (left) A portion of stage 5 showing the 3 bottleneck tiles in black, and possible horizontal and vertical growth from the top bottleneck tile.  (middle and right)  Stages 3 and 4.  The black tile is the top left bottleneck tile, the green locations are those which correctly match the smaller stage, and the red are those which go out of bounds of $\Ufrac$. Clearly, all tiles in green positions will be able to grow, and then erroneous growth is forced to occur immediately east of the green tiles, where no other tiles could prevent this growth. (Note that only a single tile needs to be placed in a red location to break the shape of $\Ufrac$.)}
\label{fig:U5-bad-growth-in-3-and-4}
\end{center}
\end{figure}

\paragraph {\bf Case2: } By observation of the shape of $\Ufrac$, we can see that for stage $j > 1$, the address of the middle left bottleneck tile is $d^{j-2}bd$.  Assume that some middle-placed tile of the left-middle of $u_j$ attaches in a location which is east of the level $j-3$ sub-stage in which it is contained.  Due to the similar geometric surroundings of the middle left bottleneck, using the same reasoning as for Case1, we know that it is impossible for a middle-placed tile to attach in such a location.  Also similar to Case1, we use this fact to imply that no middle-placed tile can attach in a location lower than the middle left bottleneck.  However, we also use it to imply that no middle-placed tile can attach in a location which is at a $y$-coordinate which is only $2$ lower than the top left bottleneck.  This is obviously true by inspection of $\Ufrac$, since the only way that middle-placed tiles could grow so far north would also be to grow outside of (to the east of) the level $j-3$ sub-stage.  (By translating the colored portions down to the middle bottleneck tile, Figure~\ref{fig:U5-bad-growth-in-3-and-4} can also depict how this would occur.)

\paragraph {\bf Case3: } Assume a bottom-placed tile of stage $u_j$ appears in a location with the same $y$-coordinate as the middle-left bottleneck.  Similar to Case3 of the proof for Theorem~\ref{thm:H-impossible}, we will sketch how such growth in stage $j$ actually grows out of bounds of $\Ufrac$ when it occurs in stage $k > j$.  Essentially, since the address of the bottom left bottleneck in stage $j$ is $d^{j-2}cd$ and inspection of $\Ufrac$ shows that the only paths in the left-middle of $u_j$ which would allow bottom-placed tiles to attach the necessary distance northward would rely on growth up through the north boundary of the level $j-2$ substage containing that bottleneck tile, when that growth occurs in stage $k > j$, then it grows out the north side of one of the sub-stages smaller than the level $k-2$ sub-stage, which is the only one that has neighboring locations to the north that are still contained within $\Ufrac$.  (See Figure~\ref{fig:U-bottom-bad} for an example.)  Therefore, bottom-placed tiles in $u_j$ cannot grow to the same vertical height as the middle left bottleneck.

\begin{figure}[htp]
\begin{center}
\includegraphics[width=4.0in]{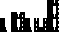}
\caption{(left) A portion of $u_4$ showing the 3 left bottleneck tiles (black) and the possible upward growth (green) from the bottom left bottleneck tile which could grow to positions horizontal to the middle left bottleneck, (middle) A portion of the assembly of $u_5$ depicting similar growth (yellow), (right) A portion of the assembly of $u_6$ showing how that bottom-placed growth from stages 4 or 5 would go out of bounds of $\Ufrac$ in stage 6.  The black tile is the bottom left bottleneck tile. The green tiles show locations which would grow out of bounds of $\Ufrac$ if the middle-placed tiles from stage $u_4$ were allowed to grow in stage $u_6$, and the yellow for the growth of stage $u_5$.}
\label{fig:U-bottom-bad}
\end{center}
\end{figure}

\paragraph {\bf Case4: }  We'll refer to the region of the left-middle of stage $j$ which falls vertically between the top left bottleneck and the middle left bottleneck, but to the east of the level $j-3$ substage containing the middle left bottleneck, as the \emph{unfilled-region}.  From Case1 we know that tiles cannot grow down solely from the top bottleneck into this region.  From Case2 we know that tiles cannot grow into this region solely from the middle bottleneck.  From Case3 we know they can't come solely from the bottom bottleneck.  Since $u_j \in S_L$, we know that the unfilled-region also didn't receive tile growth through the center position, so the only option left is for tiles to be placed there via cooperation between tiles whose ancestors were different bottleneck tiles.  Such cooperative growth clearly cannot occur between descendants of the top and middle bottleneck tiles, because it was also shown in Case2 that middle-placed tiles cannot grow to within a distance of $2$ of the vertical position of the top bottleneck, meaning there must be at least a two-tile-wide gap, preventing cooperative tile attachment.  Thus, the only remaining option is for cooperative tile attachments between middle-placed and bottom-placed tiles.  We now note that Case2 also showed that middle-placed tiles cannot grow eastward beyond the boundary of the level $j-3$ substage containing the middle bottleneck tile, and that Case3 showed that bottom-placed tiles can grow to a vertical position of at most one less than that of the middle bottleneck.  (See Figure~\ref{fig:U-bad-cooperation} for a visual depiction.)

\begin{figure}[htp]
\begin{center}
\includegraphics[width=3.0in]{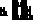}
\caption{(left) A portion of $u_4$ showing the middle and bottom bottleneck tiles (black), in green is the right side of the level $4-3=1$ substage containing the middle bottleneck which is the farthest east that middle-placed tiles can grow, and the possible bottom-placed tile growth reaching the maximum height.  The red location shows the only location where such tiles could cooperate to place another tile. (right) The same scenario is depicted for stage $u_5$.}
\label{fig:U-bad-cooperation}
\end{center}
\end{figure}

In order for a tile to be placed cooperatively between a middle-placed tile and a bottom-placed tile in $u_j$, a tile from each has to be adjacent to the same location.  There is exactly one such position, which is directly east of the middle bottleneck tile and the first which is outside of the level $j-2$ substage containing it.  (This position is above that of a possibly bottom-placed tile.) For this to occur, a middle-placed tile must reach the farthest east location of the level $j-3$ substage containing the middle bottleneck tile.  However, if we allow that same growth to occur in stage $i < j$, then it exceeds the eastern edge of the level $i-3$ substage containing the middle bottleneck tile of $u_i$, which we know cannot happen by Case2.  Thus, this is a contradiction, meaning that Case4 is also impossible.

Since none of the growth in cases $1-4$ are possible, it is impossible for the middle-left of $u_j$ to be fully tiled, and thus $\calT$ does not finitely self-assemble $\Ufrac$.

\qed

\end{proof} 

\section{Details of Finitely Self-assembling $\bf{U}$ in the 2HAM}\label{sec:construction-appendix}

In this section we give the more detail about the 2HAM system which finitely self-assembles $\bf{U}$. We start by describing the self-assembly of stage-$n$ ladder supertiles for $n\geq 3$.

\subsection{Stage-$n$ ladder supertiles for $n\geq 3$}\label{sec:stage-n-appendix}

\begin{figure}[htp!]
\centering
	   \includegraphics[width=3.5in]{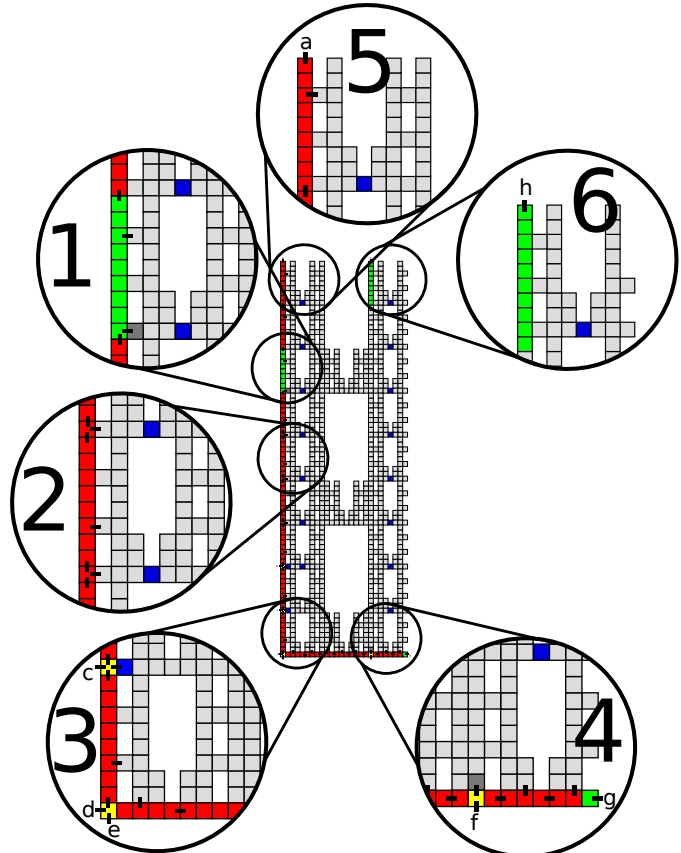}
   \caption{A schematic depiction of a grouted stage-$3$ supertile. Note the similarity between the pattern of glues labeled here and the glues of the second figure from the left in Figure~\ref{fig:grout-sequence1}. Many of the glues not depicted here are strength-$2$ glues which are hard-coded to allow either grout supertiles to self-assemble, stage-$2$ ladder supertiles to self-assemble, or rung supertiles to self-assemble. Glues depicted as strength-$1$ glues are intended to indicate how grout supertiles cooperatively bind. Glue labels shown here are for reference purposes only and do not correspond to the label in the definition of the tile set for $\calT_{\bf{U}}$.} \label{fig:grout-sequence}
\end{figure}

\begin{figure}[htp!]
\centering
	   \includegraphics[width=5in]{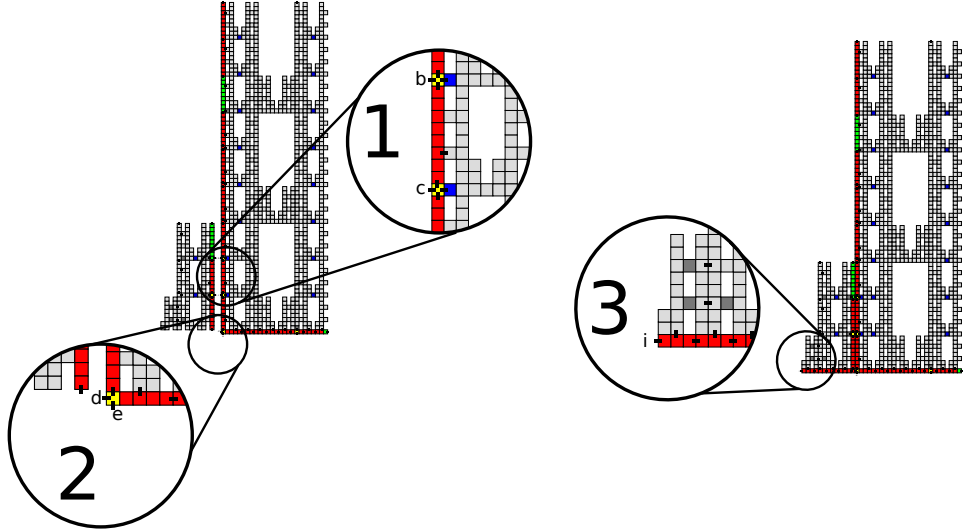}
   \caption{A schematic depiction of a grouted stage-$3$ ladder supertile and a stage-$3$ rung supertiles before (left) and after (right) they bind. Note the similarity between the pattern of glue labels here and the glue labels in Figure~\ref{fig:stage3rung}. Glue labels shown here are for reference purposes only and do not correspond to the label in the definition of the tile set for $\calT_{\bf{U}}$.} \label{fig:grout-sequence23}
\end{figure}

In Section~\ref{sec:grout-supertiles} we saw that tile types can be defined to self-assemble base supertiles and grout supertiles such that there is an assembly sequences where these supertiles bind to yield stage-$3$ ladder supertiles. Moreover, the stage-$3$ ladder supertiles which self-assemble contain tiles with edges that contain indicating glues that are identical to the indicating glues to one of the stage-$2$ ladder supertile types, giving $10$ types of stage-$3$ ladder supertiles.
For $n\geq 3$, we note that copies of the same grout supertiles which bind to stage-$2$ ladder and rung supertiles can bind to stage-$(n-1)$ ladder supertiles, yielding grouted stage-$(n-1)$ supertiles such that appropriate grouted stage-$(n-1)$ supertiles can bind to yield a stage-$n$ ladder supertile. Moreover,  the stage-$n$ ladder supertiles which self-assemble contain tiles with edges that contain indicating glues that are identical to the indicating glues to one of the stage-$(n-1)$ ladder supertile types, and thus identical to indicating glues of one of the stage-$2$ ladder supertiles. In this section, we will describe this in more detail by describing how grout supertiles bind to a stage-$3$ ladder supertile with type $r[ef]b$.

Let $Q$ be a stage-$3$ ladder supertile with type $r[ef]b$ and let $L$ once again denote a stage-$2$ ladder supertile with type $r[ef]b$. $L$ is depicted as the second supertile from the left in Figure~\ref{fig:grout-sequence1}. We denote the supertile that is the result of grout binding to $L$ until no more grout supertiles can bind the $L'$.  We refer to the glue labels for the glues shown Figure~\ref{fig:grout-sequence}. The glues labeled $a$ through $h$ are identical and serve the same purpose as the glues labeled $a$ through $h$ belonging to edges of tiles in $L$.
Referring to the enumerated subfigures of the supertiles in Figure~\ref{fig:grout-sequence} we see in Subfigure 1, a grout supertile initially binding to $L$. Then, grout supertiles cooperatively bind one at a time to self-assemble a single tile wide column of tiles to the north and south of the grout supertile that initially binds. Subfigure 2 depicts a grout supertile cooperative binding to the south of a grout supertile that is bound to the south of the grout supertile that initially bound to $L$. Note the number of tile locations between glues $a$ and $h$, and between $e$ and $f$. Also note that Subfigure 6 depicts a helper glue, labeled $h$, which belongs to an edge of a tile of a grout supertile that must bind to a stage-$2$ ladder supertile with type $r[ef]a$ as described in the previous section. Moreover, note that in Subfigure 3 an indicator glue belonging to the west edge of the blue tile allowed for the yellow tile to bind and expose the $c$ glue. This is the reason that the indicating glues where added to the one of the versions of stage-$2$ ladder supertiles with type $r[ab]c$ and $r[ef]c$ (shown as the third figure from the left in Figure~\ref{fig:grout-sequence1}). This is depicted in more detail in Figure~\ref{fig:grout-sequence23}.

Thus far, we have described the self-assembly of stage-$n$ ladder supertiles via particular assembly sequences, ignoring many other possible assembly sequences for $\calT_{\bf{U}}$ and many possible producible supertiles. In the next section, we describe how our construction ensures finite self-assembly of $\bf{U}$ despite these many possible assembly sequence and producible supertiles.

\subsection{Proper self-assembly despite nondeterminism}\label{sec:proper-assembly-appendix}

In this section, we describe how our construction avoids race conditions that if not avoided, could lead to the self-assembly of shapes that are not contained in $\bf{U}$. See Figure~\ref{fig:spurious-growth} for an example of such a race condition.
\begin{figure}[htp]
\begin{center}
\includegraphics[width=3.0in]{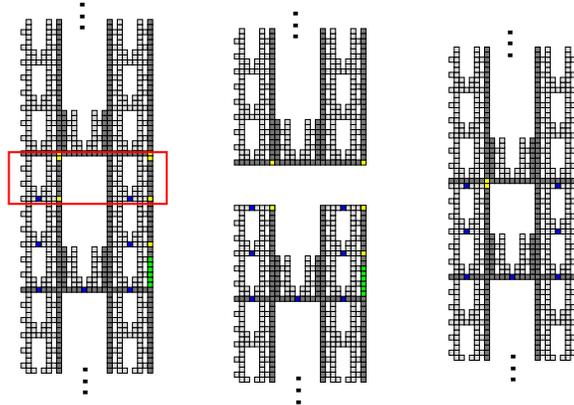}
\caption{Left: two grouted stage-$3$ ladder supertiles correctly bind via two strength-$1$ glues exposed on the north edges of the second to northernmost yellow tiles. The same glues that permit this binding may also be on north edges of tiles belonging to grout supertiles that have previously bound. For example, assume that these same glues are on north edges of the two southernmost yellow tiles. Middle: assume that the supertiles depicted here are producible. Right: erroneous binding occurs as the north glues of the southernmost yellow tiles match the south glues of the northern most yellow tiles. We show that our construction avoids such erroneous binding.}
\label{fig:spurious-growth}
\end{center}
\end{figure}
To see how our construction avoids the race condition described in Figure~\ref{fig:spurious-growth}, consider Subfigures~5 and~6 of Figure~\ref{fig:grout-sequence}. Note that in virtue of how stage-$2$ ladder supertiles self-assemble (as shown in Figure~\ref{fig:sequence-stage2ladder}) and how stage-$n$ supertiles self-assemble as a left half-ladder and right half-ladder bind, the blue tiles depicted in Subfigures~5 and~6 must bind prior to the grout supertiles with tiles that expose the $a$ and $h$ glues. Therefore, the race condition depicted in Figure~\ref{fig:spurious-growth} is avoided as the blue tile prevent the binding of $a$ and $h$ glues except for when such glues belong to the northernmost tile of a grouted stage-$n$ ladder supertile.

\begin{figure}[htp!]
\centering
	   \includegraphics[width=1.5in]{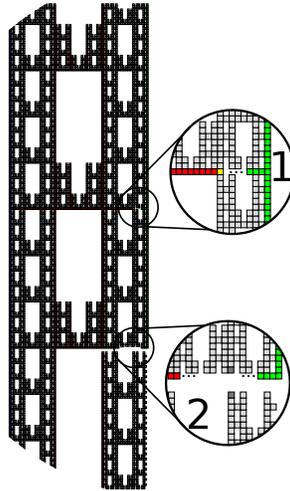}
   \caption{A depiction of grout supertiles that ``turn a corner'' to bind to south glues of tiles of a some stage-$n$ ladder supertile. Note that grout supertiles may turn a corner many times before actually binding to south glues of tiles that are southern most tiles of the stage-$n$ ladder supertile.} \label{fig:corner-turning}
\end{figure}

Subfigure 3 in Figure~\ref{fig:grout-sequence} depicts grout that initially binds to glues of south edges of tiles belonging to a stage-$3$ ladder supertile. When such binding occurs, we say that the grout supertiles ``turn a corner''. One can note that for $n\geq 4$, if grout supertiles bind to an incomplete stage-$n$ ladder supertile, it is possible for grout supertiles to turn a corner before actually reaching southernmost tiles of a stage-$n$ ladder supertile. This situation is depicted in Figure~\ref{fig:corner-turning}. In Subfigures~1 and~2, grout will turn a corner and the easternmost tile of the grout supertile that binds to a south glue of a stage-$n$ ladder supertile contains a south glue that allows a grout supertile to cooperatively attach and continue to bind to east glues of easternmost tiles of the stage-$n$ ladder supertile. In other words, grout supertiles are defined so that they are permitted to turn a corner and still continue to self-assemble a column of tiles along the side of some stage-$n$ ladder supertile. Subfigures~1 and~2 depict this situation. Subfigure~2 also depicts the special case where grout turns a corner and may bind to an indicating glue of a previous stage ladder supertile. We define grout supertiles so that in this case, appropriate binding glues will be exposed by the grout supertile that attaches to the indicating glue. It is important to note that the binding glue that is exposed on an edge of a tile belonging to the grout supertile that binds does not depend on the type of grout supertile that binds to it. Finally, in Subfigure~2, we note that we must ensure that the east edge of the easter most red tile does not contain an indicating glue as these grout supertiles are blocked by the green grout supertiles that turn a corner. To do this, we note that we can define tile types for grout supertiles so that grout supertiles can only turn a corner after such grout binds to all indicating glues (there are at most $3$ indicating glues on the left or right side of any type of stage-$n$ ladder supertile.)

\subsection{Finite self-assembly of the first quadrant}\label{sec:quadrant-appendix}

The system that has been described self-assembles higher and higher stages of the ladder supertiles. Note that $\bf{U}$, by definition, only contains points in the first quadrant of the plan. Moreover, the westernmost points (resp. southernmost points) are a vertical (resp. horizontal) line of points. We call these points the ``boundary'' of $\bf{U}$. Only self-assembling higher and higher stages of ladder supertiles would give a system that finitely self-assembles $\bf{U}$ without points on the boundary of $\bf{U}$.  In this section, we give a simple tweak the constructed 2HAM system in order to ensure that there is an assembly sequence from any producible assembly sequence to a terminal assembly with domain equal to $\bf{U}$ (including boundary points). Just as there are two versions of stage-$2$ ladder supertiles with type $r[ab]c$, this tweak involves adding one additional version of the stage-$2$ ladder supertile with type $r[ab]c$. For clarity, we say that this new version of a type of stage-$2$ ladder supertile has ``boundary type''. Figure~\ref{fig:special-base-case} depicts this ladder supertile. In addition, we add tile types which self-assemble one more type of grout supertiles. We also say that these grout supertiles have ``boundary type''. In total, there are now $11$ types of stage-$2$ ladder supertiles and $11$ types of grout.

\begin{figure}[htp!]
\centering
	   \includegraphics[width=.5in]{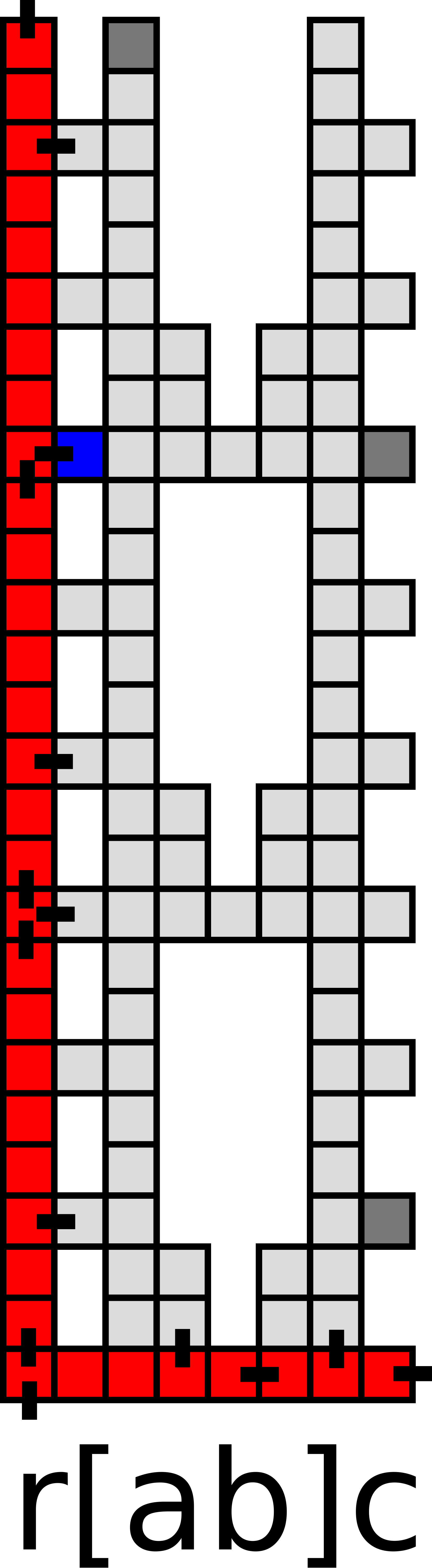}
   \caption{A depiction of a version of a stage-$2$ ladder supertile with type $r[ab]c$. The tiles of this stage-$2$ ladder supertile are shown in gray, dark gray, and blue. ``Boundary'' supertiles are shown as red tiles. } \label{fig:special-base-case}
\end{figure}

In addition, we add tile types which self-assemble supertiles that are similar to grout supertiles in that they cooperatively bind to stage-$n$ ladder supertiles with boundary type to ``partially surround'' the supertile  (by binding to glues on west or south edges of a ladder supertile). We call these supertiles \emph{boundary} supertiles. In Figure~\ref{fig:special-base-case}, the blue tile contains an indicating glue which allows a boundary supertile to initially bind. These grout supertiles bind to the stage-$2$ ladder supertile as shown and do not contain tiles with edges that contain binding glues. Moreover,  the indicating glues of a stage-$2$ ladder supertile with boundary type are defined to only permit the grouted stage-$2$ ladder supertile to have grout with boundary type. Note that this implies that if a stage-$3$ ladder supertile contain a stage-$2$ ladder supertile with boundary type, then the stage-$3$ ladder supertile must have boundary type. Higher stages are analogous to the self-assembly of stage-$3$ ladder supertiles. By adding this new type of ladder supertile and grout supertiles, we ensure that for any producible assembly in our system, there is an assembly sequence to a terminal assembly with domain equal to $\bf{U}$ (up to translation).

\subsection{U-fractal construction proof of correctness}\label{sec:ufractal-poc-appendix}

To prove that the 2HAM TAS $(T_{\bf{U}},2)$ given by our constuction finitely self-assembles $\bf{U}$ we make the following observations.

\begin{observation}\label{obs:ufractal-poc}
For any producible assembly $\alpha \in \prodasm{T_{\bf{U}}}$, there exists $n \in \N$  such that $\alpha$ is a subassembly of a stage-$n$ ladder supertile, $L$ say, and there is an assembly sequence starting from $\alpha$ with result $L$.
\end{observation}

\begin{observation}\label{obs:ufractal-poc2}
For $n\in \N$, any producible stage-$n$ ladder supertile $\alpha \in \prodasm{T_{\bf{U}}}$, there exists an assembly sequence starting from $\alpha$ that results in a stage-$(n+1)$ ladder supertile (including a ladder supertile with boundary type).
\end{observation}

\begin{observation}\label{obs:ufractal-poc3}
For each $n\geq 2$, there is a sequence of supertiles $\alpha_n$ where $\dom(\alpha_n) = U_n$ and an assembly sequence $\sigma$ from $\alpha_n$ to $\alpha_{n+1}$. Moreover, the result of $\sigma$ is $\bf{U}$.
\end{observation}

The sequence of supertiles in Observation~\ref{obs:ufractal-poc3} is a sequence of stage-$n$ ladder supertiles with boundary type. Observations~\ref{obs:ufractal-poc}, \ref{obs:ufractal-poc2}, and \ref{obs:ufractal-poc3} show that $(T_{\bf{U}},2)$ finitely self-assembles $\bf{U}$. We also note that by construction, the only terminal assemblies of $\calT_{\bf{U}}$ are the result of an assembly sequence that contains an infinite subsequence of stage-$i$ ladder supertiles for $i$ such that $2\leq i \leq \infty$.
Hence, for every stage $s\geq1$ and every terminal assembly $\alpha \in \termasm{T_{\bf{U}}}$, $\fractal{U}{s} \subset \dom(\alpha)$ (modulo translation).

\ifabstract
\ifarxiv
\newpage
\appendix

\begin{center}
	\Huge\bfseries
	Technical Appendix
\end{center}

\magicappendix \label{sec:appendix}
\fi
\fi

\end{document}